\documentclass[a4paper,11pt]{article}
\usepackage{graphicx}
\usepackage{makecell}
\usepackage[final]{changes}
\usepackage{float}
\usepackage[section]{placeins}

\usepackage{amsmath,amssymb,amsfonts}
\usepackage{amsthm}
\usepackage{mathabx}
\usepackage{hyperref}
\usepackage{float}
\usepackage{cite}
\newtheorem{optimization}{Modification}
\newtheorem{remark}{Remark}
\newtheorem{definition}{Definition}
\newtheorem{theorem}{Theorem}
\usepackage{xcolor}
\usepackage{marginnote}

\newcommand{\mygraphic}[1]{\includegraphics[height=#1]{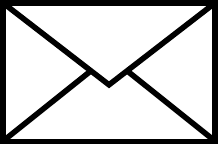}}
\newcommand{\myenv}{\raisebox{0pt}{\mygraphic{.6em}}}
\newcommand{\corresponding}[1]{#1~\myenv}
\begin{document}

\title{Multi-hop Byzantine Reliable Broadcast with Honest Dealer Made Practical
	%	\thanks{Grants or other notes
}
%\subtitle{{\color{blue}Byzantine Reliable Broadcast Made Practical on Multi-Hop Networks}
%}

\date{}

\author{Silvia Bonomi$^\star$, \corresponding{Giovanni Farina$^\dagger$$^{\star}$} , S\'ebastien Tixeuil$^\dagger$\\~\\
	$^\star$Dipartimento di Ingegneria Informatica Automatica e Gestionale\\
	``Antonio Ruberti'' (DIAG),\\ Sapienza Universit\`a di Roma,
	Rome, Italy\\
	{bonomi}$@$diag.uniroma1.it\\~\\
	$^\dagger$Sorbonne Universit\'e,\\
	CNRS, LIP6,\\
	F-75005 Paris, France\\
	Giovanni.Farina$@$lip6.fr, Sebastien.Tixeuil$@$lip6.fr}

\maketitle

\begin{abstract}
	We revisit Byzantine tolerant reliable broadcast \added{with honest dealer} algorithms in multi-hop networks. To tolerate \deleted{up to $f$} Byzantine \added{faulty} nodes \added{arbitrarily spread over the network}, previous solutions require a factorial number of messages to be sent over the network if the messages are not authenticated (e.g. digital signatures are not available). We propose \replaced{modifications}{optimizations} that preserve the safety and liveness properties of the original unauthenticated protocols, while highly decreasing their observed message complexity when simulated on several classes of graph topologies, \added{potentially opening to their employment}.
	
\end{abstract}

\section*{Introduction}
Designing dependable and secure distributed systems and networks, that are able to cope with various types of adversaries (ranging from simple errors to internal or external attackers), requires to integrate those risks from the very early design stages.
The most general attack model in a distributed setting is the Byzantine model, where a subset of system participants may behave arbitrarily (including malicious), while the rest of participants remains correct. Also, reliable communication primitives are a core building block of any distributed system.

We consider the \emph{reliable broadcast \added{with honest dealer}} problem in presence of Byzantine failures \emph{i.e.}, the problem of distributing information from a \emph{source} to every other process considering that a subset of nodes may act arbitrarily.
The reliable broadcast primitive is expected to provide two guarantees: \emph{(i)} \emph{safety} \emph{i.e.}, every message $m$ delivered by a correct node has been previously sent by the source and \emph{(ii)} \emph{liveness} \emph{i.e.}, every message $m$ sent by a correct source is eventually delivered by every correct node.

We are interested in solving this problem in a \emph{multi-hop} network, in which nodes are not directly connected to every other (\emph{i.e.} the network is not complete). In particular, nodes may have to rely on intermediate ones (hops) in order to communicate, forwarding messages till the final destination.
In case the entire system is correct the solution to reliable broadcast is trivial: every node has just to forward the received messages to all of its other neighbors (or it has to question a routing table to know which is the next node to route a specific message) and if the network is connected then it is possible for every node to communicate with every other. 
Contrarily, if just one single node is faulty, specifically Byzantine faulty, two problems may arise: (\emph{i}) messages can be modified or generated by faulty nodes that pretend the messages were sent from another node (\emph{ii}) and messages can be blocked preventing nodes to communicate. It follows that a more sophisticated protocol has to be put in place to ensure the correct communication between the parties.

Lastly, we are interested in unauthenticated solutions, namely in protocols where messages cannot be directly authenticated (\emph{e.g.} employing digital signatures) and thus the nodes cannot immediately verify that a specific received message has been previously sent by a specific other node.

\added{The reliable broadcast with honest dealer enables to simulate a completely connected distributed system equipped with reliable and authenticated channels. It follows that all the solutions designed for completely connected distributed system (Byzantine Agreement, Byzantine Reliable Broadcast etc.) can be directly deployed on top of a multi-hop network once the reliable broadcast service has been deployed.}
\subsection*{Related Works}
The necessary and sufficient condition to solve the reliable broadcast \added{with honest dealer} problem on general networks has been identified by Dolev \cite{DBLP:conf/focs/Dolev81}, demonstrating that it can be solved if and only if the network is $2f+1$-connected, where $f$ is the maximum number of Byzantine \added{faulty} nodes. 

Subsequently, research efforts followed three paths: \emph{(i)} replacing global conditions with local conditions, \emph{(ii)} employing cryptographic primitives, or \emph{(iii)} considering weaker broadcast specifications.

The Certified Propagation Algorithm (CPA) \cite{pelc2005broadcasting} is a protocol that solves reliable broadcast in networks where the number of Byzantine nodes is \emph{locally bounded}, \emph{i.e.}, in any given neighborhood, at most $f$ processes can be Byzantine.
This algorithm has been later extended \cite{pagourtzis2017reliable} along several directions: (\emph{i}) considering different thresholds for each neighborhood, (\emph{ii}) considering additional knowledge about the network topology, and (\emph{iii}) considering the general adversary model.

The Byzantine tolerant reliable broadcast \added{with honest dealer} can also be addressed employing cryptography (\emph{e.g.}, digital signatures) \cite{DBLP:conf/dsn/DrabkinFS05, DBLP:conf/osdi/CastroL99} that enables all nodes to exchange messages guaranteeing authentication and integrity (authenticated protocols). 
The main advantage is that the problem can be solved with simpler solutions and weaker conditions (in terms of connectivity requirements). 
However, on the negative side, most of those solutions rely on a third party that handles and guarantees the cryptographic keys, thus  the safety of those protocols is bounded to the crypto-system (a potential single point of failure).

Lastly, the broadcast problem has been considered weakening safety and/or liveness property \emph{e.g.}, allowing to a (small) part of correct processes to either deliver fake messages, or to never deliver a valid message~\cite{maurer2014byzantine,maurer2015containing,maurer2016tolerating}.

Let us note that a common assumption considered by Byzantine tolerant reliable broadcast protocols is to use authenticated point-to-point channels, which prevents a process from impersonating \replaced{several others}{several ones} (Sybil attack) \cite{DBLP:conf/iptps/Douceur02}.
The real difference between cryptographic (authenticated) and non-cryptographic (unauthenticated) protocols for reliable broadcast is how the cryptography is employed: non-cryptographic protocols, in fact, may use digital signatures just within neighbors for authentication purposes, whereas the cryptographic protocols make use of cryptographic primitives to enable the message verification even between non-directly connected nodes. 
Let us finally remark that an authenticated channel not necessarily requires the use of cryptography \cite{DBLP:journals/wc/ZengGM10}.

\par
Although the Byzantine tolerant reliable broadcast problem \added{with honest dealer} has been extensively studied considering alternative and additional assumptions, the solution provided by Dolev \cite{DBLP:conf/focs/Dolev81} is the only one for general settings and it has never been revisited from a performance perspective. 
Indeed, this solution hints at poor scalability since it requires a factorial number of copies (with respect the size of the network) of the same message to be spread and potentially verified in order to be accepted. \deleted{by a correct node and let think that} \added{This suggest that} solving reliable broadcast in the weakest system model (\emph{i.e.}, Dolev's solution \cite{DBLP:conf/focs/Dolev81}) is practically infeasible.

\subsection*{Contributions}
We review and improve previous solutions for reliable broadcast in multi-hop networks, where at most $f$ nodes can be Byzantine faulty, making no further assumption with respect to the original setting~\cite{DBLP:conf/focs/Dolev81}.	
In more details, (\emph{i}) we propose and evaluate \added{modifications to the state of art protocols} that preserve both safety and liveness properties of the original algorithms, and (\emph{ii}) we define message selection \replaced{policies}{techniques} in order to prevent Byzantine faulty nodes from flooding the network and to reduce the total number of messages exchanged.

In a preliminary work \cite{DBLP:conf/ladc/bonomi18}, we focused on random network topologies, we defined two \replaced{modifications to the state of art protocols}{optimizations that we present in the sequel}, we proposed one preliminary message selection technique, and we carried out a performance analysis in the scenario where all processes are correct.
In this work, by extensive simulations for variously shaped networks and considering active Byzantine processes spreading spurious messages over the network, we show that our modifications enable to keep the message complexity close to quadratic (in the size of the network). Our work thus paves the way for the practical use of Byzantine tolerant reliable broadcast solutions in realistic-size networks. 

\section*{System Model and Problem Statement}
\subsection*{System Model}
We consider a distributed system composed by a set of $n$ processes $\Pi=\{p_1, p_2, \dots p_n\}$, each one having a unique integer identifier. 
Processes are arranged in a communication network. This network can be seen as an undirected \deleted{and not complete} graph $G = (V,E)$ where each node represents a process $p_i \in \Pi$ (\emph{i.e.} $V = \Pi$) and each edge represents a communication channel connecting two of them $p_i, p_j \in \Pi$ (\emph{i.e.} $E \subset \Pi \bigtimes \Pi$), such that $p_i$ and $p_j$ can communicate. In the following, we interchangeably use terms \emph{process} and \emph{node} and we refer to \emph{edges, links} and \emph{communication channels} similarly.

We assume an omniscient adversary able to control up to $f$ processes allowing them to behave arbitrarily. We call them \emph{Byzantine} processes. 
All the processes that are not Byzantine faulty are \emph{correct}. Correct processes do not a priori know the subset of Byzantine processes.
Processes have \emph{no global knowledge about the system} (\emph{i.e.} the size or the topology of the network) except the value of $f$.

Communication channels allow connected processes to exchange messages, providing two interfaces: \emph{SEND}$(p_{dest},m)$ and \emph{RECEIVE}$(p_{source},m)$.
The former requests to send a message $m$ to process $p_{dest}$, the latter delivers the message $m$ sent by process $p_{source}$. Processes that are not linked with a communication channel have to rely on others that relay their messages in order to communicate, in a \emph{multi-hop} fashion.
We assume \emph{reliable} and \emph{authenticated} communication channels, which provide the following guarantees: \emph{(i)} \emph{reliable delivery}, namely if a correct process sends a message $m$ to a correct process $q$, then $q$ eventually receives $m$; \emph{(ii)} \emph{authentication}, namely if a correct process $q$ receives a message $m$ with sender $p$, then $m$ was previously sent to $q$ by $p$.

We consider a \emph{synchronous system}, namely we assume that \emph{(i)} there is a known upper bound on the message transmission delays 
and \emph{(ii)} a known upper bound on the processing delays.
We assume a computation that evolves in sequential synchronous \emph{rounds}. Every round is divided in three phases: \emph{(1)} send, where processes send all the messages for the current round, \emph{(2)} receive, where processes receive all the messages sent at the beginning of the current round and \emph{(3)} computation, where processes execute the computation required by the specific protocol.
\added{In a single round, any message can traverse exactly one hop, namely the message exchange occurs only between neighbor processes.}
We measure the time in terms of the number of rounds.
\subsection*{Problem Statement}
We consider the \emph{reliable broadcast \added{with honest dealer}} problem from a \emph{source} $s$ assuming $f$ Byzantine failures arbitrary spread in the network \cite{DBLP:conf/focs/Dolev81}.
A protocol solves the Byzantine tolerant Reliable Broadcast (BRB) \added{with honest dealer} problem if the following conditions are met: 
\begin{itemize}
	\item \emph{(Safety)} if a \emph{correct} process delivers a message $m$, then it has been previously sent by the source; 
	\item \emph{(Liveness)} if a \emph{correct} source broadcast a message $m$, then $m$ is be eventually delivered by every correct process.
\end{itemize}
Notice that in the case of a correct source, all correct processes deliver the broadcast message.
Instead, if the source is Byzantine faulty, then not all correct processes necessarily deliver the broadcast message and/or they may deliver different messages.

We referred (following the literature) with message $m$ to a content (\emph{i.e.} a value) that has been broadcast. 
Every information spreading protocol place a content inside a message with a specific format, adding protocol related overhead. Therefore, to ease of explanation, we refer with \emph{content} to the value that has been broadcast and with \emph{message} to the one exchanged by a protocol. Therefore, a message refers to the union of a content and the overhead added by the employed protocols.

\section*{Background}
We start by presenting and remarking some definitions and theoretical results to lead the reader in an easier understanding. Subsequently, we present the state-of-the-art protocols solving the BRB problem and we provide an analysis of them.
\subsection*{Basic Definitions and Fundamental Results}
For all the definitions and results that follow, let us consider the cube graph $\hat{G}$ depicted in Figure \ref{fig:cube} as example.
\begin{figure}[!htb]
	\centering
	\includegraphics[width=.65\textwidth]{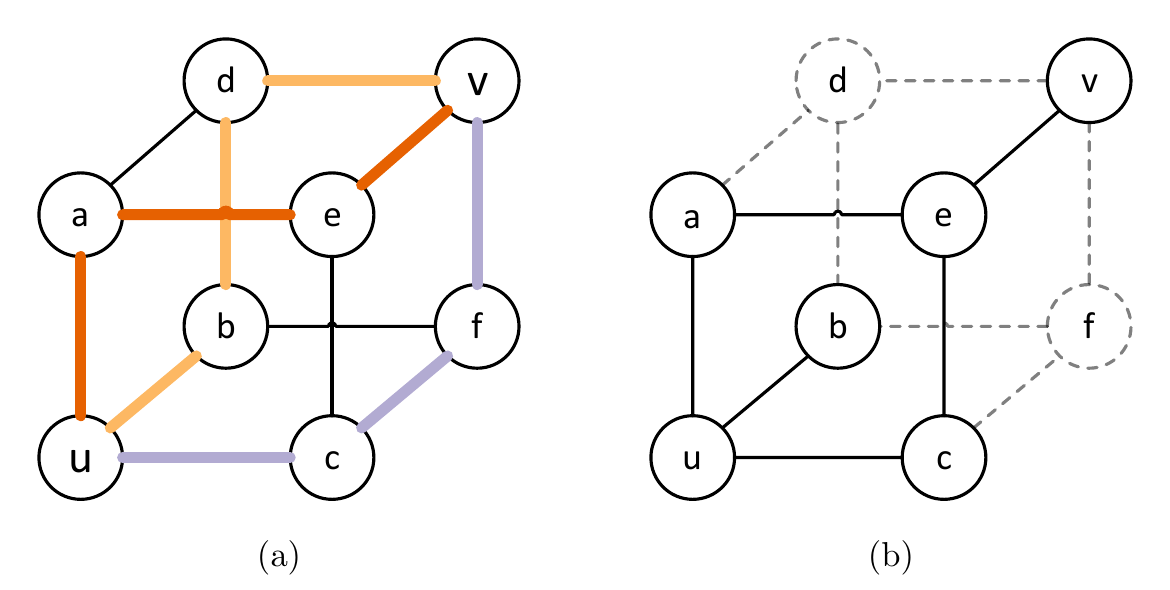}
	\caption{\textbf{Disjoint Paths (a) and Vertex Cut (b)} of a cube graph.}
	\label{fig:cube}
\end{figure}
\begin{definition}[neighbors]
	Given an undirected graph $G = (V,E)$, two nodes $u,v$ are \emph{adjacent} or \emph{neighbors} if there is an edge connecting them (\emph{i.e.} $\{u,v\} \in E$).
\end{definition}
In graph $\hat{G}$, the neighbors of node $u$ are nodes $a,b$ and $c$.
\begin{definition}[path]
	Given an undirected graph $G = (V,E)$, a \emph{path} is a sequence of adjacent nodes without repetitions (\emph{i.e.} $path := (v_1,v_2,\dots,v_x) : \forall i \in \replaced{[1,x-1]}{[1,x]}, ~ v_i \in V, ~ \{v_i,v_{i+1}\} \in E$). 
	The two extreme nodes of a path are called \emph{ends}.
\end{definition}
\begin{definition}[connected nodes and connected graph]
	Given an undirected graph $G = (V,E)$, two nodes $u,v$ are \emph{connected} if there exist at least one path with ends $u,v$, they are \emph{disconnected} otherwise. The graph $G$ is connected if it exist at least one path between every pair of nodes.
\end{definition}
In graph $\hat{G}$, the sequence $(u,a,e,v)$ is a path with ends $u$ and $v$, thus nodes $u$ and $v$ are connected.
\begin{definition}[\replaced{independent}{indipendent}/disjoint paths]
	Given an undirected graph $G = (V,E)$, two or more of its paths are \emph{independent} or \emph{disjoint} if they share no node except their ends. 
\end{definition}
In graph $\hat{G}$, the sequence $(u,c,f,v)$ is another path that is disjoint with respect to $(u,a,e,v)$.
\begin{definition}[vertex cut]
	Given an undirected graph $G = (V,E)$, the removal of a set of nodes $C \subset V$ from $G$ results in a subgraph $G_C = (V_C,E_C)$, where $V_C = V -  C$ and $E_C \subset E : \forall \{v_i,v_j\} \in E, \{v_i,v_j\} \in E_C \iff v_i,v_j \notin C$.
	Given two not adjacent nodes $u,v \in V$, a \emph{vertex cut} $C \subset V - \{u,v\}$ for $u,v$ is a set of nodes whose removal from the graph $G$ disconnects $u$ from $v$, namely $u,v$ are disconnected in $G_C$. 
\end{definition}
In graph $\hat{G}$, the set of nodes $\{d,e,f\}$ is a vertex cut \replaced{for $(u,v)$, because its removal will}{, indeed it} disconnects \replaced{nodes $u$ from $v$}{node $v$ from the other nodes}.

Given an undirected graph $G = (V,E)$ and two not adjacent nodes $u,v$, the maximum number of mutually disjoint paths with ends $u$ and $v$ is referred with $\kappa'(u, v)$, and the size of the smallest vertex cut $C$ separating $u$ from $v$ is referred with $\kappa(u, v)$.
\begin{remark} [Global Menger Theorem]
	\label{rem:gmt}
	A graph is $k$-connected (or it has \emph{vertex connectivity} equals to $k$) if and only if it contains $k$ independent paths between any two vertices.
\end{remark}
\begin{remark} [Vertex Cut VS Disjoint Paths]
	\label{rem:vcmd}
	Let $G = (V,E)$ be a graph and $u,v \in V$. Then the minimum number of vertexes that disconnects $u$ from $v$ in $G$ is equal to the maximum number of disjoint $u-v$ paths in $G$, namely  $\kappa(u,v) = \kappa'(u, v)$.
\end{remark}
In graph $\hat{G}$, the maximum number of disjoint paths between nodes $u,v$ ($\kappa'(u, v)$) is 3, as depicted in Figure  \ref{fig:cube}a. Furthermore, at least 3 nodes have to be removed from the network in order to disconnect nodes $u,v$ (Figure \ref{fig:cube}b), thus showing the equivalence $\kappa'(u, v)= \kappa(u, v)$.

The reader can refer to~\cite{Diestel_2017} for addition details.

The Byzantine Reliable Broadcast problem can be solved under the assumptions we considered in the system model when the following condition is met:
\begin{remark} [Condition for Byzantine Reliable Broadcast \cite{DBLP:conf/focs/Dolev81}]
	\label{rem:cnes}
	Given a network $G$ composed of $n$ processes where at most $f$ can be Byzantine faulty, the Byzantine Reliable Broadcast can be achieved if and only if the vertex connectivity of $G$ is at least $2f+1$.
\end{remark}
\subsection*{Byzantine Reliable Broadcast Protocols}
There exists two solutions addressing the Byzantine Reliable Broadcast \added{with honest dealer} problem in the system model we considered, that are the Dolev \cite{DBLP:conf/focs/Dolev81} and the Maurer et al. \cite{DBLP:conf/srds/MaurerTD15} algorithms. Any other solutions for the BRB problem makes extra or different assumptions (\emph{e.g.} digital signatures, higher density networks, weaker versions of safety or liveness, etc.).

The protocols that follow are defined by:
\begin{itemize}
	\item a \emph{propagation algorithm}, which rules how messages are spread over the network;
	\item a \emph{verification algorithm}, that decides if a content can be accepted by a process guaranteeing the safety of reliable broadcast. 
\end{itemize}

The basic idea behind the following protocols is to leverage the authenticated channels to collect the labels of processes traversed by a content, in order to compute the maximum disjoint paths, \replaced{in the case of Dolev}{in the first following}, or the minimum vertex cut, \replaced{in the case of Maurer et al.}{in the second following}, of all the paths traversed by the content. 
Those two methodologies are theoretically equivalent due to the Menger Theorem (Remark \ref{rem:vcmd}), \emph{i.e.} if a message can traverse $f+1$ disjoint paths in a network, then it can traverse paths such that their minimum vertex cut is $f+1$ and vice versa.

\subsubsection*{Dolev Reliable Broadcast Protocol \emph{(D-BRB)}}
Dolev \cite{DBLP:conf/focs/Dolev81} defined the seminal algorithm addressing the BRB problem.
The messages exchanged by his protocol have the format $m := \langle s,content,path\rangle$, where $s$ is the label of the process asserting to be the source, $content$ is the content to broadcast and $path$ is a sequence of nodes.

\medskip
\noindent Propagation Algorithm:

\begin{enumerate}
	\item the source process $s$ sends the message $m = \langle s,content,\emptyset \rangle$ to all of its neighbors;
	\item a correct process $p$ saves and relays any message ${m = \langle s, content, path_i \rangle}$ sent by a neighbor $q$ to all of other neighbors not included in $path_i$, appending to $path_i$ the label of the sender $q$, namely process $p$ stores and multicasts ${m = \langle s, content, path_i \cup \{q\} \rangle}$. 
	
	The messages carrying $path_i$ with loops or $path_i$ that includes the label of the receiving process $p$ are discarded at reception.
\end{enumerate}

\medskip
\noindent Verification Algorithm:
\begin{enumerate}
	\item given a set of messages $m_i = \langle s,content,path_i \rangle$ received by process $p$ and carrying the same values for $s$ and $content$, the $content$ is delivered by $p$ if there exist $f+1$ disjoint paths among of all the related $path_i$.
\end{enumerate}

\subsubsection*{Maurer et al. Reliable Broadcast Protocol \emph{(MTD-BRB)}}

Maurer et al. \cite{DBLP:conf/srds/MaurerTD15} extended and improved the algorithm defined by Dolev to deal with dynamic distributed systems, where the communication network changes over the time. As a matter of facts, a static communication network (our case) can be seen as a dynamic network that never changes, making their solution employable on the system model we are considering.

The format of messages exchanged by their protocol is \\$m:= \langle s,content,pathset \rangle$, again carrying the information about the process $s$ asserting to be the source and the $content$ of broadcast. The difference with respect the previous algorithm is the data structure employed to collect the labels of traversed nodes: $pathset$, that it discards the traversing order.  

Furthermore, \emph{MTD-BRB} verifies the minimum vertex cut of the pathset traversed by a content \replaced{instead of}{with respect} the maximum disjoint paths. The reason is that on dynamic networks the Menger theorem (Remark \ref{rem:cnes}) does not hold, specifically the minimum vertex cut may be greater than or equal to the maximum disjoint paths between two nodes \cite{DBLP:journals/jcss/KempeKK02}. 

%\medskip
\newpage
\noindent Propagation Algorithm:
\begin{enumerate}
	\item the source process $s$ sends the $m = \langle s,content,\emptyset \rangle$ to all of its neighbors;
	\item a correct process $p$ saves and relays any message ${m = \langle s, content, pathset_i \rangle}$ sent by a neighbor $q$ to all of other neighbors not included in $pathset_i$, appending to $pathset_i$ the label of the sender $q$, namely process $p$ stores and multicasts ${m = \langle s, content, pathset_i \cup \{q\} \rangle}$. 
	
	The messages carrying $pathset_i$ with duplicates or $pathset_i$ that includes the label of the receiving process $p$ are discarded at reception.
\end{enumerate}

\medskip
\noindent Verification Algorithm:
\begin{enumerate}
	\item given a set of messages $m_i = \langle s,content,pathset_i \rangle$ received by process $p$ and carrying the same values of $s$ and $content$, the $content$ is delivered by $p$ if the minimum vertex cut of the related $pathset_i$ is at least $f+1$.
\end{enumerate}

\subsection*{Discussion} 
The two protocols we presented are \emph{quiescent}, namely at a certain point all correct processes stop sending messages. To be precise, this is true only if all processes are correct or at least the additional knowledge about the size of the system is given to the processes to guarantee the termination of message spreading. Contrarily, a Byzantine process may continuously introduce spurious messages carrying $\tilde{path_i/pathset_i} = \{random\_label\}$ that are forwarded to all processes. We temporarily assume, for ease of evaluation, that all processes are correct in the analysis that follows.
In order to evaluate and compare the protocols reported and the solution we are going to design, we analyze the following metrics:
\begin{enumerate}
	\item \emph{message complexity} \emph{i.e.}, the total number of messages exchanged in a single broadcast (the amount of messages exchanged from the beginning of the broadcast till the moment when all processes stop sending messages);
	\item \emph{delivery computational complexity} \emph{i.e.}, the complexity of the procedure executed by a process during the computation phase to decide if a content can be accepted.
	\item \emph{broadcast latency} \emph{i.e.}, the time between the beginning of the broadcast and the time when all correct processes deliver the content.
\end{enumerate}
The message complexity of \emph{D-BRB} protocol is factorial in the size of the network. The reason is that for every path connecting the source with any other node (\emph{i.e.} that are order of the permutations over the full set of nodes) a message with related $path_i$ is generated.
This potentially results in an factorial number of $path_i$ to elaborate by every process in order to deliver a single content. 
\added{For sake of explaination, let us consider the cube graph depicted in Figure \ref{fig:cube} and let us assume that process $u$ starts a broacast, thus spreading a content with an empty $path$ that will traverse the paths $(u,a)$, $(u,b)$ and $(u,c)$ on the communication network. The neighbors of the source will receive the content and its related (empty) path, they will attach the label of the sender and they will forward it to all of their neighbors not already included, e.g. process $a$ will forwards the content with the path $(u)$ to processes $d$ and $e$ and thus a message related to the paths $(u,a,d)$ and $(u,a,e)$ will be generated (and the same will be done also by processes $b$ and $c$). The process $d$ will receive the path $(u)$ from $a$ and $b$ (the same happens for processes $e$ and $f$ from different processes). Consequently, a message carring $(u,a)$ will be forwarded by $d$ to $b$ and $v$ and a message carrying $(u,b)$ will be sent by $d$ to $a$ and $v$. The messages continue to be generated as long as all possible paths are traversed, one message for each path.}

Furthermore, to the best of our knowledge, the only method available to identify $f+1$ disjoint $path_i$ is the reduction to a NP-Complete problem, \emph{Set Packing} \cite{Garey:1990:CIG:574848}. We refer to this method with \emph{DP} (Disjoint Paths), namely to the reduction and solution of the associated Set Packing instance. This implies that the delivery complexity of the algorithm is exponential.

The \emph{D-BRB} guarantees the safety and liveness properties of BRB when the strict enabling condition is met (Remark \ref{rem:cnes}), respectively because the Byzantine processes $b_1, b_2, \dots b_f$ cannot propagate a different content $\tilde{content} \neq content$ with source $s$ through no more than $f$ disjoint paths, and assuming a vertex cut of size $f$ made by the faulty processes, $f+1$ disjoint paths are still available between any pairs of correct processes.

The \emph{MTD-BRB} protocol is equivalent with respect the message complexity and delivery complexity to \emph{D-BRB}. Specifically, even if all the $paths_i$ over the same set of nodes are all collapsed in a single $pathset$, they are still factorial in the number of nodes (\emph{i.e.} they are order of the combinations over the full set of nodes), and messages carrying any possible $pathset_i$ are generated, potentially leading to an input of factorial size for the verification algorithm.

Again, to the best of our knowledge, the only method available to identify a vertex cut of size less than or equal to $f$ is the reduction to a NP-Complete problem, \emph{Hitting Set} \cite{Garey:1990:CIG:574848}. We refer to this method as \emph{VC} (Vertex Cut), namely to the reduction and solution of the associated Hitting Set instance.
This implies that the delivery complexity of this algorithm is exponential too.

The safety and liveness properties of BRB are guaranteed by \emph{MTD-BRB} due to the same argumentation made for \emph{D-BRB}: the Byzantine processes $b_1, b_2, \dots b_f$ cannot propagate a different content $\tilde{content} \neq content$ with source $s$ through pathsets with a minimum vertex cut greater than $f$ and they cannot make a vertex cut on the communication network greater than $f$.

The broadcast latency of both protocols is bounded by the graph metric called \emph{wide diameter} \cite{hsu1994container}. \added{Given a $k$-connected graph G,} the wide diameter is the maximum number $l$ such that there exist $k$ internally disjoint $(u, v)$-paths in \deleted{a graph} $G$ of length at most $l$ between any \replaced{pair of}{two distinct} vertices $u$ and $v$. This value depends on the graph topology. In the worst case, the wide diameter of a graph is \replaced{$n-k$}{$n-k+1$} \cite{DBLP:journals/dm/HsuL94}. It follows that the broadcast latency of both protocols is upper bounded by \replaced{$n-k$}{$n-k+1$}, because in at most $n-k$ rounds $k$ disjoint path are traversed between every pair of nodes. \added{As a clarifying example, let us consider a $k$-connected generlized wheel graph with $n$ nodes, that is composed by the disjoint union between a cycle and a $k-2$ clique (a grafical example is depicted in Figure \ref{fig:alltopologies}b) and let us chose as source a node on the cycle and let us focus on one of its neighbors on the cycle at distance two. It is possible to verify that in order to interconnect the pair of nodes we are considering through $k$ disjoint paths, one path of length $n-k$ has to be traversed.}

The Table \ref{table:metrics1} summarizes the presented analysis.

\begin{table}[]
	\centering
	\bgroup
	\def\arraystretch{1.5}%
	\begin{tabular}{l|c|c|}
		\cline{2-3}
		\multicolumn{1}{c|}{}                     & Dolev & Maurer et al. \\ \hline
		\multicolumn{1}{|l|}{Message Complexity}  & \makecell{Factorial \\ (Permutations of nodes)}     & \makecell{Factorial \\ (Combinations of nodes)}           \\ \hline
		\multicolumn{1}{|l|}{Delivery Complexity} & \replaced{NP}{EXPTIME}     & \replaced{NP}{EXPTIME}            \\ \hline
		\multicolumn{1}{|l|}{Broadcast Latency}   & $\leq \replaced{n-k}{n-k+1}$     & $ \leq \replaced{n-k}{n-k+1}$             \\ \hline
	\end{tabular}
	\egroup
	\caption{Analysis of state of art protocols}
	\label{table:metrics1}
\end{table}
\section*{Practical Reliable Broadcast Protocol}
Due to the high message complexity and delivery computational complexity of the reviewed protocols, they do no scale and they cannot be successfully employed. 

We further analyze some deeper details of the aforementioned protocols and we define simple modifications that result in drastically reducing the message complexity.
Specifically, we start by arguing that \emph{pathsets} and \emph{VC} should be preferred respectively as message format and verification algorithm. Subsequently, we propose \replaced{modifications}{optimizations} that aim in reducing the amount of messages spread by preventing from forwarding useless messages, thus redefining a protocol solving BRB.

\subsection*{Paths VS Pathsets}
It is possible to note that, given the solutions available, there is no reason to prefer path over pathset while collecting the label of the traversed processes,
indeed: \emph{(i)} due to the reduction to set related problems, paths are converted into sets to be analyzed; \emph{(ii)} two paths over the same set of nodes are not disjoint and have a cut of size equal to 1, and \emph{(iii)} the pathsets interconnecting two endpoints are fewer than the relative paths. For those reasons we adopt the pathset data structure as message format to collect the labels of traversed processes in designing an improved protocol.
\subsection*{Minimum Vertex Cut VS Maximum Disjoint Paths}
We remark that both verification algorithms \replaced{solve a}{reduce to} NP-Complete problems and, considering the Menger theorem in Remark \ref{rem:vcmd}, one may conclude that there is no tangible reason to prefer one among VC and DP.
As a matter of fact, the equivalence between the two metrics in Remark \ref{rem:vcmd} occurs when no restriction on the length of the paths is assumed. \replaced{In fact,  when the path lenght is bounded,}{Contrarily,} the minimum vertex cut between two nodes may be higher than or equal to the maximum number of disjoint paths interconnecting them \cite{Lovsz1978}.
Let us take the example proposed in Figure~\ref{fig:minlessmaxx}~\cite{Lovsz1978}, let us focus on nodes $u$ and $v$ as endpoints and consider only the paths of length at most $5$. It can be verified that at least two nodes have to be removed from the graph in order to disconnect $u$ from $v$ considering only the paths of length at most $5$. \replaced{Nevertheless}{Instead}, no two disjoint paths exist considering only the paths with the same constraint. 
In other words, given a graph $G$ of $n$ nodes and considering  only the paths of length at most $l<n$, the size of the minimum vertex cut of those interconnecting two nodes may be greater than or equal to the maximum number of disjoint paths interconnecting them.
This implies that, whenever a synchronous system is assumed and the paths are all traversed synchronously like in our system model (\emph{i.e.} the path of length 1 are all traversed in 1 instant (round), the paths of length 2 are all traversed in two instant (round) and so on), it may be possible to interconnect two endpoints with a minimum cut equal to $k$ in fewer hops (\emph{i.e.} rounds) with respect $k$ disjoint paths. 
This results also in saving in message complexity if a halting condition is embedded inside the protocol, namely if the message propagation stops when all correct processes delivered the content.
For this reason we adopt VC as verification methodology.

\begin{figure}
	\centering
	\includegraphics[width=.6\textwidth]{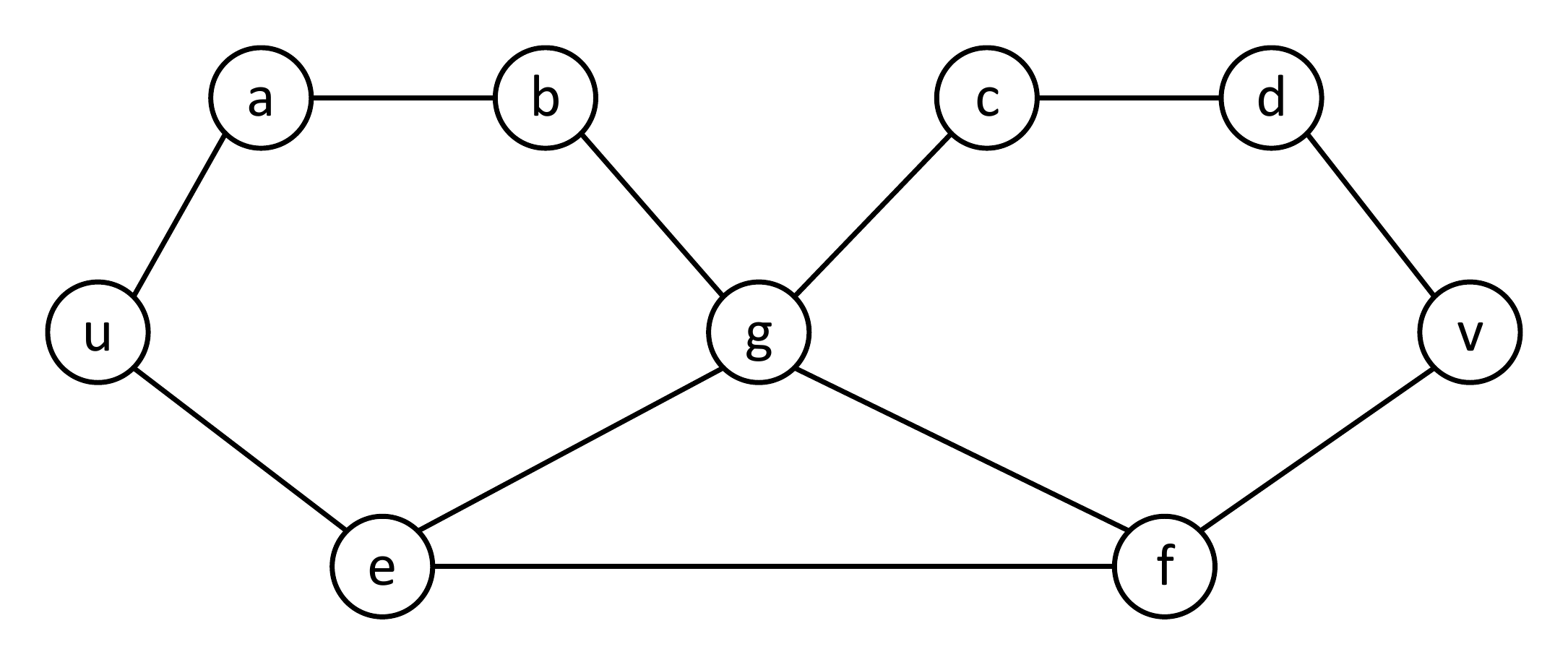}
	\caption{\textbf{Graph example to compare minimum vertex cut and maximum disjoint paths.}}
	\label{fig:minlessmaxx}
\end{figure}

\subsection*{Improvements}
\subsubsection*{Practical Reliable Broadcast Protocol (BFT-BRB)}
We redefine a protocol for the Byzantine Reliable Broadcast \added{with honest dealer}. This protocol employs the same message format and verification algorithm of \emph{MTD-BRB}, namely the label of the processes traversed by a message are collected in \emph{pathsets} and the contents are verified through the VC methodology. We introduce four \replaced{modifications}{optimizations} in the propagation algorithm and one in the verification algorithm that aim to reduce the total number of messages exchanged and we prove their correctness, namely that their employment do not prevent the original algorithms of Dolev and Maurer et al. from enforcing safety and liveness of Byzantine Reliable Broadcast \added{with honest dealer} when the strict enabling condition is met (Remark \ref{rem:cnes}), because they prevent from forwarding messages that are not useful for the delivery of a content.

\begin{optimization}
	\label{optcpa}
	If a process $p$ receives a content directly from the source $s$ (\emph{i.e.} the source and the sender coincides), then it is directly delivered by $p$.
\end{optimization}

\begin{optimization}
	\label{optempty}
	If a process $p$ has delivered a $content$, then it can discard all the related $pathsets$ and relay the $content$ only with an empty pathset to all of its neighbors.
\end{optimization}
\begin{optimization}
	\label{optneighdel}
	A process $p$ relays pathsets related to a content only to the neighbors that have not yet delivered it.
\end{optimization}
\begin{optimization}
	\label{optdiscard}
	If a process $p$ receives \replaced{a content with an empty pathset}{an empty pathset related to a content} from a neighbor $q$, then $p$ can discard from relaying and analyzing any further pathset related to the content that contains the label of $q$.
\end{optimization}
\begin{optimization}
	\label{optstop}
	A process $p$ stops relaying further pathsets related to a content after \replaced{it}{that the it} has been delivered and the empty pathset has been forwarded.
\end{optimization}

The modification \ref{optcpa} \replaced{follows from the defitions of disjoint paths and vertex cut, indeed a path of lenght two is disjoint with respect every other one with the same ends, and the vertex cut is defined between not adjacent nodes, thus there is no vertex cut between neighbors.}{potentially reduces broadcast latency by one round. Indeed, only the neighbors of the source are affected by it}. 

The purpose of modifications \ref{optempty}, \ref{optneighdel}, and \ref{optdiscard} is to reduce the amount of messages exchanged by the protocol \added{and to be analyzed by processes}. The modification \ref{optempty} also provides a transparent way to get the neighbors $q$ of a process $p$ know that a specific content has been delivered by $p$. This one has already been employed~\cite{nesterenko2009discovering} for the purpose of topology reconstruction.

The modification \ref{optstop} introduces a halting condition in the protocol with respect the state of art, indeed all correct processes stop from relaying further messages at the round subsequent the last delivery of a process. 
Furthermore, this modifications makes the original solutions quiescent without assuming that processes know the size of system.

\added{Let us consider the network topology depicted in Figure \ref{fig:cube}a as an example to detail the advantages introduced by the presented modifications. Let us select node $u$ as source process and let us consider the all the paths of lenght two starting from $u$, namely $(u,c,f), (u,c,e), (u,b,f), (u,b,d), (u,a,d), (u,a,e)$. 

Processes $d$, $e$ and $f$, following modification 1 will relay only an empty path instead of extending the paths they received, namely avoiding to generate $(u,c,f,v)$, $(u,c,e,v)$, $(u,b,f,v)$, $(u,b,d,v)$, $(u,a,d,v)$, $(u,a,e,v)$.

Processes $d$, $e$ and $f$, leveraging modification 1, know that the nodes $a,b$ and $c$ have already delivered the content associated to the paths. Applying modification 2, processes $d$, $e$ and $f$ do not relay further paths to $a$, $b$ and $c$, namely they do not generate paths $(u,c,f,b)$, $(u,c,e,a)$, $(u,b,f,c)$, $(u,b,d,a)$, $(u,a,d,b)$, $(u,a,e,c)$.

The modification 3 applies in cases a process $p$ receives paths in round $r_i$ but it delivers the associated content in a round $r_j > r_i$. A neighbor $q$ of $p$ that has not yet delivered the content will get the extension of paths received by $p$ in $r_i$ and potentially the empty path in $r_j + 1$. The modification $3$ enables $q$ to discard from the analysis in delivering the associated content all paths previously received from $p$ and to consider only the empty pathset.}

The pseudo code of our protocol is presented in Figure \ref{fig:algo}. For the ease of explanation and notation, we show the procedure and variables only related to the broadcast of a single $content$ spread by $s$. 

Initially, every process is not aware about the nodes in its neighborhood but it can easily retrieve them with authenticated channels.
For every not delivered content, a process stores \emph{(i)} the received pathsets related to the content (\emph{Pathsets} variable), \emph{(ii)} the pathsets not yet relayed (\emph{To\_Forward} variable) and \emph{(iii)} the labels of neighbors that have delivered the content (\emph{Neigh\_Del} variable).  

Every process starts the round with the send phase, namely selecting the messages to forward and transmitting them. In particular, it extracts part or all of the message related to a content to relay (\emph{select} function), and it forwards them to all of its neighbors that have not yet delivered the content, thus applying modification \ref{optneighdel} in line $8$.

During the receive phase, for every received message related to a content not yet delivered, the label of the sender is attached to the received \emph{pathset} and the resulting collection is stored in order to be considered for the delivery and to be forwarded (we assume an implicit mechanism avoiding duplicate \emph{pathsets}). The modification \ref{optempty} enables a process $p$ to know that a sender $q$ has delivered the content (line $15$). Then, the modification \ref{optneighdel} allows $p$ to discard part of the $pathset$ previously received (lines $17-22$) and that may arrive (line $13$). 

Finally, in the computation phase, all received \emph{pathsets} related to the content are analyzed. Specifically, in case a process has received the content directly from the source $s$ (\emph{i.e.} the sender and the source coincides, the receiver has received the $pathset$ $\{s\}$, line $26$, modification \ref{optcpa}) or the minimum vertex cut computed on \emph{Pathsets} is greater than $f$, it delivers the content, it discards all the \emph{pathsets} not yet forwarded and it enqueues to relay the empty \emph{pathset}, namely applying modification \ref{optempty} in lines $30,31$.

The implementation of modification \ref{optstop} can be \replaced{found}{find} in line $12$. Indeed, once that a process has delivered the content, it discards all the residual $pathsets$ to forward (line $30$). In the receive phase all the messages related to the content already delivered are discarded (line $12$) due to modification \ref{optdiscard}, thus the select function in line $6$ only extracts the empty pathset in the round subsequent the delivery.

\begin{figure}
	\centering
	\includegraphics[width=\textwidth]{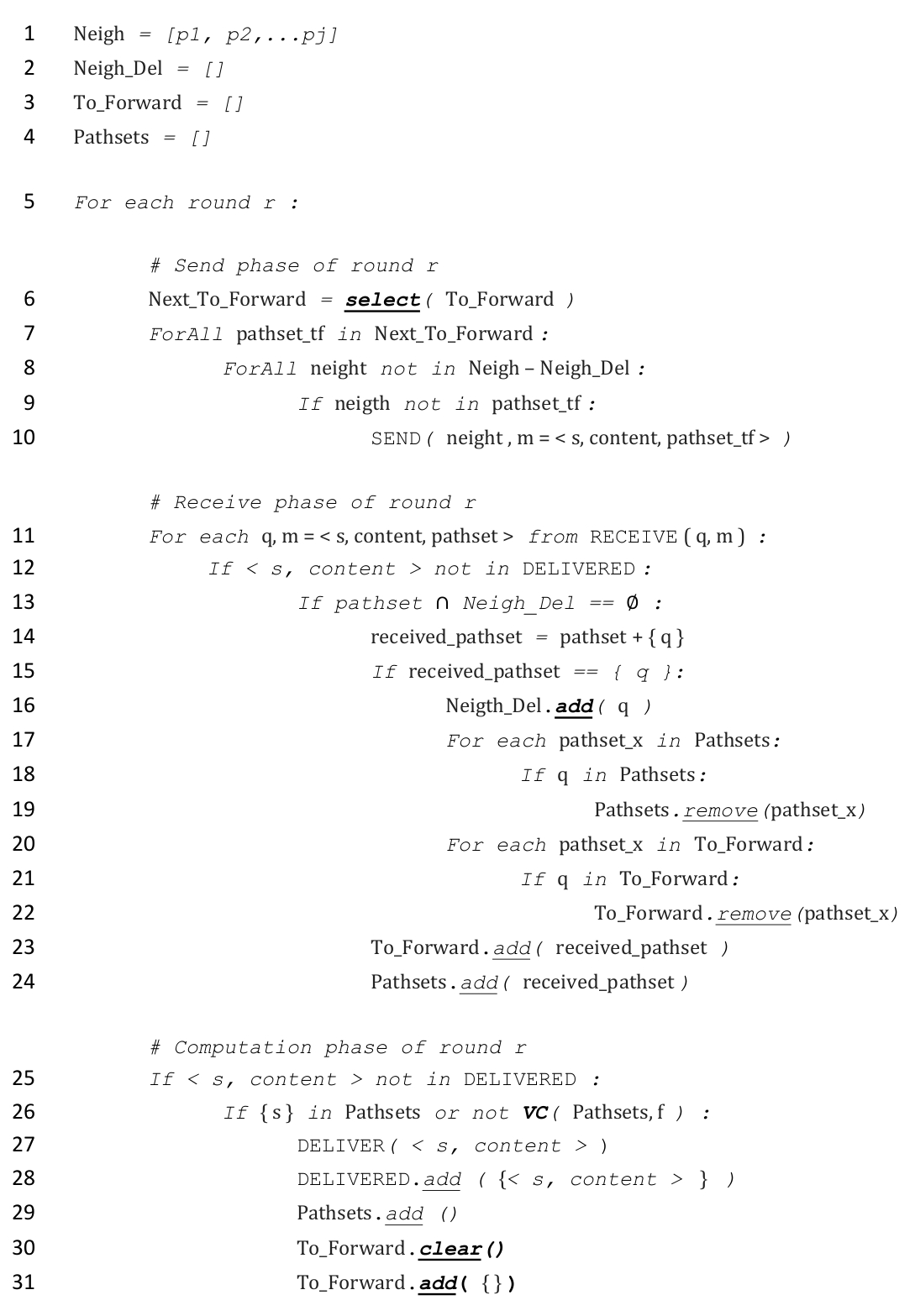}
	\caption{\textbf{Byzantine Tolerant Reliable Broadcast \added{with honest dealer} Protocol.}}
	\label{fig:algo}
\end{figure}

We prove the correctness of the proposed modifications thought the following theorems (assuming the system model we presented and under the assumption of the strict condition in Remark \ref{rem:cnes}).
\begin{theorem}
	\label{th:sourcedel}
	Let $p$ be a process executing either the Dolev or the Maurer \emph{et al.} algorithm to broadcast a content. If $p$ delivers a content received directly from the source, then the safety property continues to be satisfied (\emph{i.e.} employing modification \ref{optcpa}).
\end{theorem}
\begin{proof}
	It follows directly from the property of the channels (reliable and authenticated), indeed the channels guarantee that every received message has been previously sent by the sender, that coincides with the safety property of reliable broadcast.
\end{proof}

\begin{theorem}
	\label{th:safe_retrasm}
	Let $p$ be a process executing either the Dolev or the Maurer \emph{et al.} algorithm to broadcast a content. If $p$ delivered a content, then $p$ can relay that content with an empty path/pathset and the safety property continues to be satisfied (\emph{i.e.} employing modification \ref{optempty}).	
\end{theorem}

\begin{proof}
	The aim of the information about the nodes traversed by a content is to enable a process $p$ to decide whether it can be safely accepted. Once it has been delivered, the information about the nodes traversed before reaching $p$ is not useful, because the content has been already verified as safe by $p$.
\end{proof}

\begin{theorem}
	\label{th:safe_ignore}
	Let $p$ be a process executing either the Dolev or the Maurer \emph{et al.} algorithm to broadcast a content and let us assume that the modification \ref{optempty} is employed. Even if $p$ does not relay messages carrying the content to its neighbors that already delivered it, the liveness property continues to be satisfied (\emph{i.e.} employing modification \ref{optneighdel}).
\end{theorem}

\begin{proof}
	Let us assume that there exists three processes $p,q,r$ such that only $q$ has already delivered the content and that, among others, the following communication channels are available: $(p,q)$, $(q,r)$.
	From Theorem \ref{th:safe_retrasm} we know that process $q$ can safely relay the content with an empty path/pathset (\emph{i.e.} employing modification \ref{optempty}). Thus, any further path/pathset containing $p$ and $q$, after the delivery of $q$, does not affect the results of DP and VC verifying the content on $r$. It follows that any further transmission related to the content from $p$ to $q$ can be avoided after that $q$ has delivered without compromising liveness.
\end{proof}

\begin{theorem}
	Let $p$ be a process executing either the Dolev or the Maurer \emph{et al.} algorithm to broadcast a content and let us assume that the modification \ref{optempty} is employed. If process $p$ receives an empty path/pathset relative to a content from a neighbor $q$, than $p$ can discard from its analysis and from relaying further path/pathset containing the label of $q$ and the liveness property continues to be satisfied (\emph{i.e.} employing modification \ref{optdiscard}).
\end{theorem}
\begin{proof}
	Let us assume that there exists three processes $p,q,r$ such that only $p$ has already delivered a content and that, among others, the following communication channels are available: $(p,q)$, $(p,r)$. We have to prove that process $p$ can discard, verifying the associated content, further path/pathset containing the label of $q$ but $\{q\}$ without affecting the liveness property. This follows from the fact that path/pathset of unit lenght are included in every solution of the VC and DP and that any path/pathset containing more labels does not increase the value computed by VC and DP.
	We have to prove that this reasoning extends also for process $r$, so that process $p$ can avoid relaying further path/pathset over $\{q\}$. On process $r$, any path/pathset that extends $\{q,p\}$ does not increase the value obtained by VC and DP. It follows that any other path/pathset over $\{q\}$ has not to be relayed.	
\end{proof}

\begin{theorem}
	\label{th:halting}
	Let $p$ be a process executing either Dolev or Maurer \emph{et al.} algorithm to broadcast a content and let us assume that the modification \ref{optempty} is employed. If $p$ has delivered and relayed the content with an empty path/pathset to all of its neighbors, then $p$ can stop from relaying further related paths/pathsets and the liveness property continues to be satisfied (\emph{i.e.} employing modification \ref{optstop}).
\end{theorem}

\begin{proof}
	It follows from the fact that any further path/pathset related to the content received and relayed by $p$ does not increase the minimum cut/the maximum disjoint paths computed on other processes with respect the empty path/pathset relayed by $p$. Said differently, all the neighbors of $p$ receive the paths/pathset $\{p\}$ and any further path/pathset relayed by $p$ becomes $\{\dots, p\}$, increasing neither the minimum vertex cut nor the maximum disjoint paths.
\end{proof}

\subsubsection*{Preventing Flooding and Forwarding Policies.}
We highlighted the fact that the verification algorithm has potentially to analyze a factorial, in the size of the network, amount of pathsets even only considering all processes to be correct. Nevertheless, a Byzantine process $b$ can potentially flood the network with spurious messages (\emph{i.e.}, $\tilde{m} := \langle \tilde{s}, \tilde{content}, \tilde{pathset}\rangle$ where $\tilde{s}$, $\tilde{content}$ and $\tilde{pathset}$ can be invented by the faulty process) that are also diffused by the correct ones. 
Considering that the amount of messages plays a crucial impact on the employment of the protocol we defined, a countermeasure must be researched. 

A common way to limit the flooding capability of Byzantine processes is to constraint the channel capacity of every process, namely limiting the amount of messages that every process is allowed to send in a time window.

Noticed that, by introducing such a constraint, we are limiting the relaying capability of every process, while the Byzantine processes can continuously generate spurious messages potentially preventing the liveness property to be satisfied. It follows that a selection policy among all the messages to relay is demanded.

Every process has to relay \emph{pathsets} to all of its neighbors that have not yet delivered the content.
A \emph{pathset} that may lead a neighbor $q$ to the delivery of the associated content has not to contain the label of $q$ (because it \replaced{would be}{is} directly discarded), namely a process $p$ has to select among the \emph{pathsets} to forward the ones that do not include the label of $q$. 
There may be many pathsets that do not include $q$. Thus, we consider and evaluate \replaced{two}{three} selection policies: \emph{(i)} \emph{Multi-Random} and \emph{(ii)} \emph{Multi-Shortest} \deleted{and \emph{(iii)} \emph{Shortest-One-For-Every}}. 
The Multi-Random is an extension of the forwarding policy proposed in \cite{DBLP:conf/ladc/bonomi18}.

The algorithms for the pathsets selection implementing the Multi-Random and Multi-Shortest policies are presented in Figure \ref{fig:multipol}. The selection iteratively picks one \emph{pathset} and checks if it is ``useful'' for any neighbor (\emph{i.e.} if any neighbor to contact is not included in the \emph{pathset}). This selection continues till \emph{(i)} all the neighbors to contact receives at least one \emph{pathset} where they are not included or \emph{(ii)} the bound on channel capacity has been reached.
The Multi-Random policy iteratively picks randomly a possible \emph{pathset} to forward, the Multi-Shortest gives priority to the shorter ones.

\deleted{The algorithm implementing the Shortest-One-For-Every policy is presented in Figure \ref{fig:ofepol}. The selection iteratively picks one \emph{pathset} for every neighbor that has not yet delivered the content, giving priority to the shorter pathsets.} 

We compare and analyzed them both in the following.

\begin{figure}
	\centering
	\includegraphics[width=.85\textwidth]{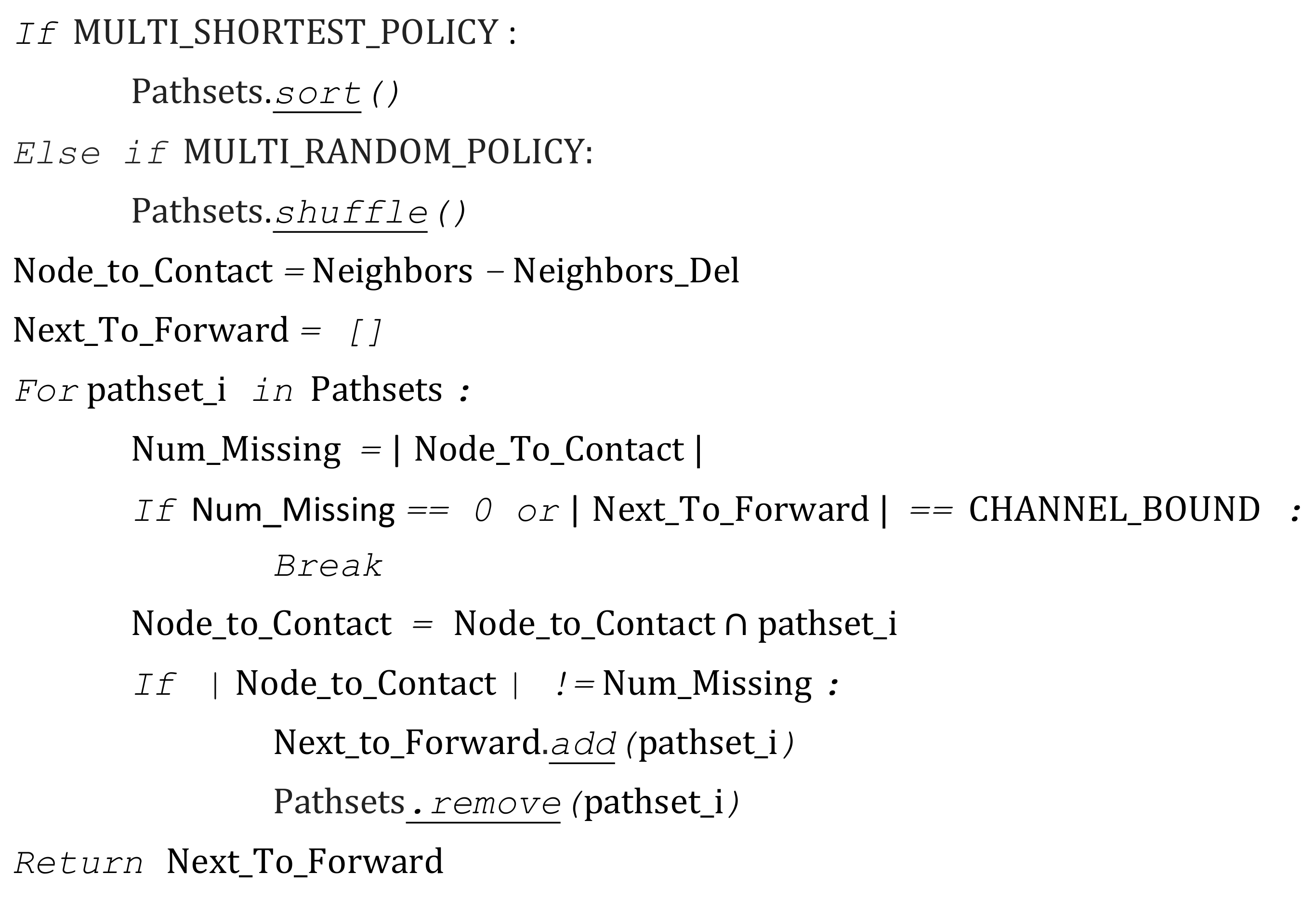}
	\caption{\textbf{Multi-Shortest and Multi-Random Policies.}}
	\label{fig:multipol}
\end{figure}

\section*{Practical Reliable Broadcast Evaluation}
We simulate the protocol and the policies we proposed in order to evaluate their effectiveness and to compare our protocol with the state-of-art solutions.

According to the system model we defined, we simulate single broadcasts that evolves in rounds. Therefore, the passage of time is measured in number of rounds. 

We made use of the implementation provided by Gainer-Dewar and Vera-Licona \cite{gainer2017minimal} for the algorithm defined by Murakami and Uno \cite{murakami2014efficient} to solve the VC reduction to the hitting set problem.

We consider the following parameters in our simulation:
\begin{itemize}
	\item $n$, \emph{i.e.} the size of the network considered;
	\item $k$, \emph{i.e.} the vertex connectivity of the network considered;
	\item \emph{topology}, \emph{i.e.} the topology of network considered;
	\item \emph{channel capacity}, \emph{i.e.} the maximum number of messages that a process can send in a link per round;
	\item \emph{kind of failure}, \emph{i.e.} how faulty process behave;
	\item \emph{forwarding policy}, \emph{i.e.} one among \emph{Multi-Shortest} and \emph{Multi-Random} \deleted{and \emph{One-For-Every} presented above}.	
\end{itemize}
In order to carry an analysis as complete as possible, we consider the following network topologies:
\begin{itemize}
	\item $k$-regular $k$-connected random graph \cite{DBLP:journals/cpc/StegerW99};
	\item k-pasted-tree \cite{DBLP:journals/tcs/BaldoniBQP09};
	\item k-diamond \cite{DBLP:journals/tcs/BaldoniBQP09};
	\item multipartite wheel;
	\item generalized wheel \cite{Xu:2010:TSA:1965102};
	\item Barab\'asi-Albert graph \cite{barabasi1999emergence}.
\end{itemize} 

\begin{figure}[!htb]
	\centering
	\includegraphics[width=\textwidth]{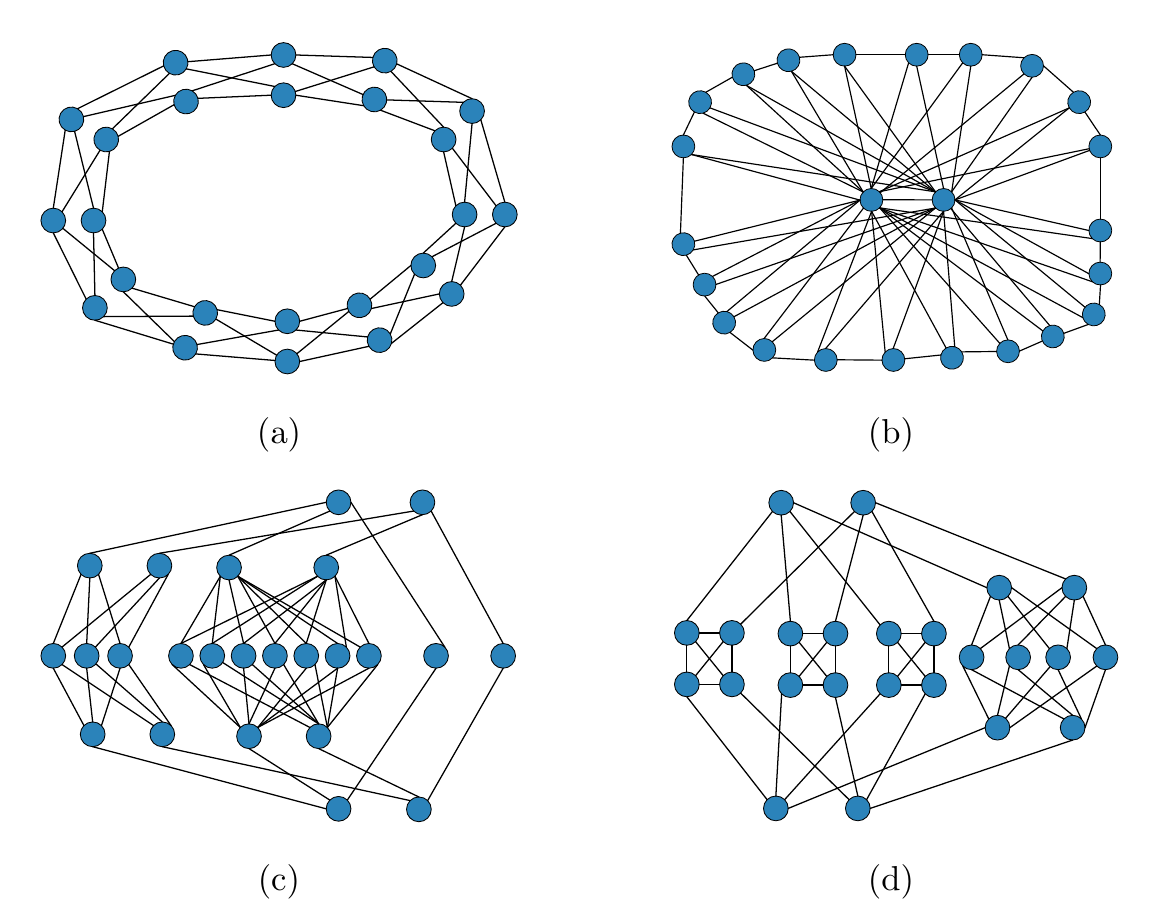}
	\caption{\textbf{Graph topologies}. $n=24$, $k=4$. (a) multipartite wheel, (b) generalized wheel, (c) k-pasted-tree, (d) k-diamond.}
	\label{fig:alltopologies}
\end{figure}

A graph is regular if every node is connected to the same number of neighbors, namely in a $k$-regular graph every node is connected exactly to $k$ neighbors.
The $k$-regular $k$-connected graphs have vertex connectivity equals to $k$ with the minimum necessary number of edges.
The $k$-regular $k$-connected random graphs are the ones uniformly sampled among all possible regular graphs employing the sampling methodology defined in \cite{DBLP:journals/cpc/StegerW99}. 

The k-pasted-trees and k-diamond graphs are \emph{Logarithmic Harary Graph} \cite{DBLP:conf/icdcsw/JenkinsD01}, namely topologies designed to be robust to failures and suited for distributed systems where the information spreading occurs by message flooding. Indeed, they are $k$-connected graphs with a logarithmic diameter and with minimal edges guaranteeing the node-connectivity (\emph{i.e.} the removal of an edge decreases vertex connectivity of the network). For specific values of network size $n$ and vertex connectivity $k$ they are $k$-regular. A graphical example of k-pasted-trees and k-diamond are respectively presented in Figures \ref{fig:alltopologies}c and \ref{fig:alltopologies}d. 

We refer with \emph{multipartite wheel} to a regular graph composed by the concatenation of disjoint \replaced{groups}{sets} of $k/2$ nodes such that every node in a \replaced{group}{set} is connected to exactly all the $k/2$ nodes in other 2 \replaced{groups}{sets} and no node inside a \replaced{group}{set} is connected with others of the same \replaced{group}{set}. A graphical example is provided in Figure \ref{fig:alltopologies}a. 

Notice that $k$-regular $k$-connected graphs can be constructed in several ways, indeed we are considering four different constructions that are either always regular or regular for specific settings. The sequel demonstrates that the specific construction impacts protocol performance. 

We considered also the Barab\'asi-Albert graphs that models complex and social networks with \deleted{vertex connectivity following} scale-free power law \added{degree distribution}. The aim is to evaluate our protocol also on topologies not designed for distributed systems.

Finally, we consider the generalized wheel, \emph{i.e.} the topology generated by the disjoint union between a cycle and a $k-2$ clique. An example can be found in Figure \ref{fig:alltopologies}b. It has been considered as a worst case scenario. 

We \replaced{carry our simulations either considering}{consider} the maximum number of tolerable faulty processes, thus for every $k$-connected network we assume $f = \lfloor(k - 1) / 2\rfloor$ failures (Remark \ref{rem:cnes})\added{, or we tested all possible values for $f$ between 0 and $\lfloor(k - 1) / 2\rfloor$}. \replaced{In any case,}{It follows that} processes deliver a content only when the related pathsets have a minimum vertex cut greater than \replaced{$\lfloor(k - 1) / 2\rfloor$}{$f$}.

We consider two configurations for the channel capacity: \emph{bounded} and \emph{unbounded}. The former constrains processes to send a limited number of messages per link in every round, the latter imposes no restriction. For the bounded case we assume a bound for the channel capacity equal to $f+1$ messages. \deleted{The value of $f+1$ has been chosen to allow at least one node in the neighborhood of a process to receive at least one pathset not generated by a Byzantine neighbor while employing the Shortest-One-For-Every forwarding policy.}

The source code of the presented simulations can be found online \cite{sourcecode}.

\subsection*{Simulating Byzantine Behaviours}
We move from the scenario where all processes are correct, to the case where the $f$ Byzantine processes  act as crash failures, (thus not relaying any message, we refer to them as \emph{passive} Byzantines), till the case they spread spurious messages (we refer to them as \emph{active} Byzantines). Specifically to this last scenario, we have to notice that spurious contents (\emph{i.e.} contents generated by Byzantine processes $b_i \neq s$ sent inside a message with source $s$) are never accepted by correct processes (if the BRB enabling condition in Remark \ref{rem:vcmd} is met) and their spreading and verification is disjoint with respect the content broadcast by the source (because they are related to a different $s$ - $content$). For this reason we impose to Byzantine processes to spread only spurious \emph{pathsets} in our simulations (thus relaying the content broadcast by the source). The purpose is to flood the correct processes with spurious pathsets trying to not facilitate the achievement of the delivery condition. In detail, the Byzantine processes diffuse \emph{pathsets} containing the label of one correct neighbor of the receiver in the first round, and \emph{pathseths} containing one of the correct neighbor of the receiver with a random label in the subsequent rounds. Every Byzantine process sends $f+1$ messages (the maximum amount allowed by the channel capacity) containing different \emph{pathsets} on every of its link per round.

\added{Additionally, we consider two kinds of active Byzantine processes: \emph{omniscient} and \emph{general}. Omniscient active faulty processes know the content that the source is going to spread before receiving it through a message, thus they start flooding correct processes with spourious pathsets from the beginning of broadcast. General active faulty processes, instead, spreads spurious pathsets in the round subsequent the first reception of a message containing the content, namely as soon as they get knowledge about the content through the network.}
\added{Notice that there are other strategies that Byzantine processes may adopt generating spourious pathsets, especially if such Byzantines are omniscient about the state of all other processes. The Byzantine strategy that we adopted has been choosen to allows faulty processes to generate pathsests that may be selected by the correct process due to their lenght}.

In every simulation, the source and the Byzantine processes are randomly placed.

\replaced{For all the results we are going to show we directly plot all the measures we got as points (except for Figures \ref{fig:stateofart},\ref{fig:varyF_cut},\ref{fig:varyF_byz} where the mean of the measures is depicted) in order to show their distributions, and we accordingly increase the size of the points with higher density.}{the confidence interval of the registered measures, assuming a normal distribution and a confidence level of 95\%.}

\subsection*{Comparison with the state of art}
We start comparing the message complexity of the state of art solutions with our protocol. We consider k-regular k-connected random graphs, we assume the vertex connectivity $k$ equal to $3$ and $5$ and we simulate \emph{D-BRB}, \emph{MTD-BRB} and \emph{BFT-MTB} considering unbounded channels and all correct processes, and we vary the size of the network from $n=6$ to $n=20$.
\begin{figure}[!htb]
	\includegraphics[width=\textwidth]{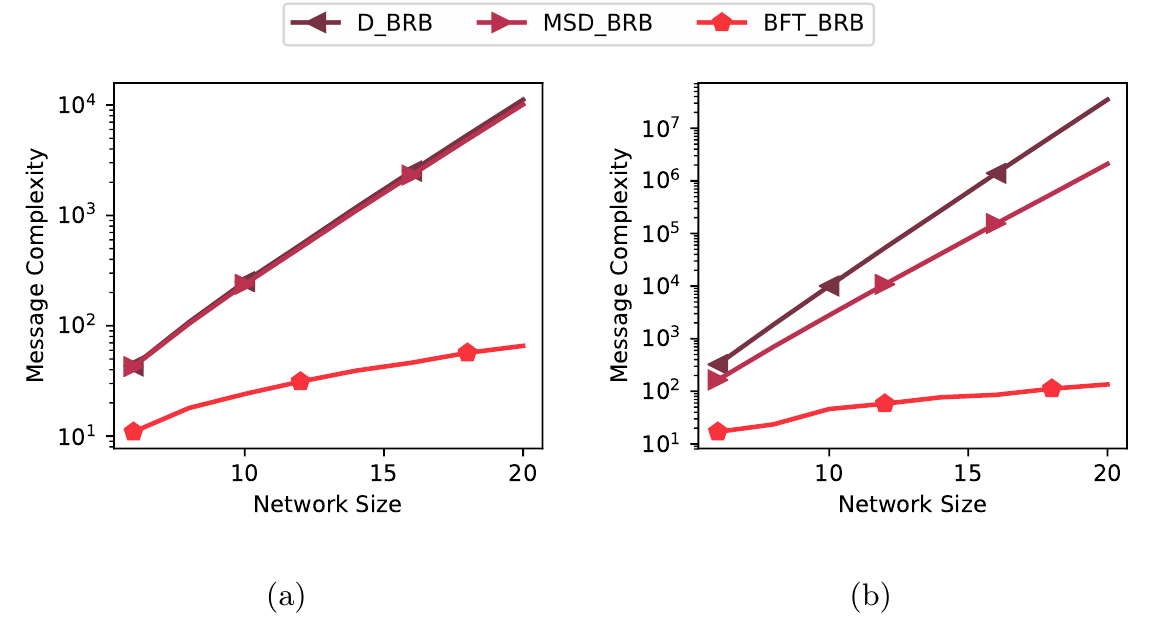}
	\caption{\textbf{State of art VS our protocol, message complexity.} Random regular network, unbounded channels, $f$ =$\lfloor(k -1)/2\rfloor$ all correct. (a) $k=3$, (b) $k=5$.}
	\label{fig:stateofart}
\end{figure}
We previously remarked \replaced{about the}{that} lack in the state of art protocols \replaced{of}{is} a halting condition, indeed they generate all source-to-other paths/pathsets in every execution. 
It can be noticed in Figure \ref{fig:stateofart} that the modifications we defined have a remarkable impact on the message complexity even in a small and all-correct scenario.
It can also be noticed the advantage gained by choosing \emph{pathsets} over \emph{paths}, as expected.

\subsection*{Multi-Random VS Multi-Shortest}
We proposed as a countermeasure against the capability of Byzantine processes to flood the network a constraint on the channel capacity, namely limiting the amount of messages that a process can send over a link per round, and we set this bound equal to $f+1$. Then, we proposed \replaced{two}{three} forwarding policies to select which pathsets relay in the actual round.
Assuming bounded channels, we compare the presented policies, \emph{Multi-Random} and \emph{Multi-Shortest} \deleted{and \emph{One-For-Every}}, considering networks of size $n=100$, topologies random regular, multipartite wheel, k-diamond and k-pasted-tree, passive \deleted{and active} Byzantine failures. The results are presented in Figures \ref{fig:multirandmsg} and \ref{fig:multishort100all} (notice that scale of the graphics in Figure \ref{fig:multirandmsg} are logarithmic). 
Starting \replaced{with}{from} the Multi-Random policy, it can be seen in Figures \ref{fig:multirandmsg} that, while for some graphs the Multi-Random policy acts smoothly, the random regular (confirming the results achieved in our preliminary work~\cite{DBLP:conf/ladc/bonomi18}) the multipartite wheel graphs \added{and the k-diamond}, there exist topologies where the broadcast latency and message complexity may conspicuously increase (k-pasted-tree). \added{It follows that on some kind of graphs the selections of paths that the Multi-Random policy may take are not equivalment with respect the protocol progression and that additional criterion have to be considered  in the selection.} This lead us to discard such a policy to be one generally employable.
\deleted{The results achieved simulating the One-For-Every and Multi-Shortest policies are plotted in Figures \ref{fig:ofemsg},\ref{fig:multishortmsg},\ref{fig:ofelat},\ref{fig:multishortlat}.  
	It can be deduced that the former is more correlated to the specific network topology. Indeed, we obtained distinguishable results while considering different networks (Figure \ref{fig:ofemsg}). Furthermore, the One-For-Every policy appears less robust against Byzantine active processes: it is possible to observe in Figure \ref{fig:ofemsg}b that the message complexity considerably increase when the Byzantine processes spread spurious messages.}
Contrarily, the performance achieved employing the Multi-Shortest policy appears  \replaced{not affected by this misbehaviour}{less correlated to the specific topology and more robust against passive Byzantine processes} (Figures \ref{fig:multishort100all}). 
\added{Therefore,} we further investigate the Multi-Shortest policy while increasing the size of the network. 
\begin{figure}[!htb]
	\centering
	\includegraphics[width=\textwidth]{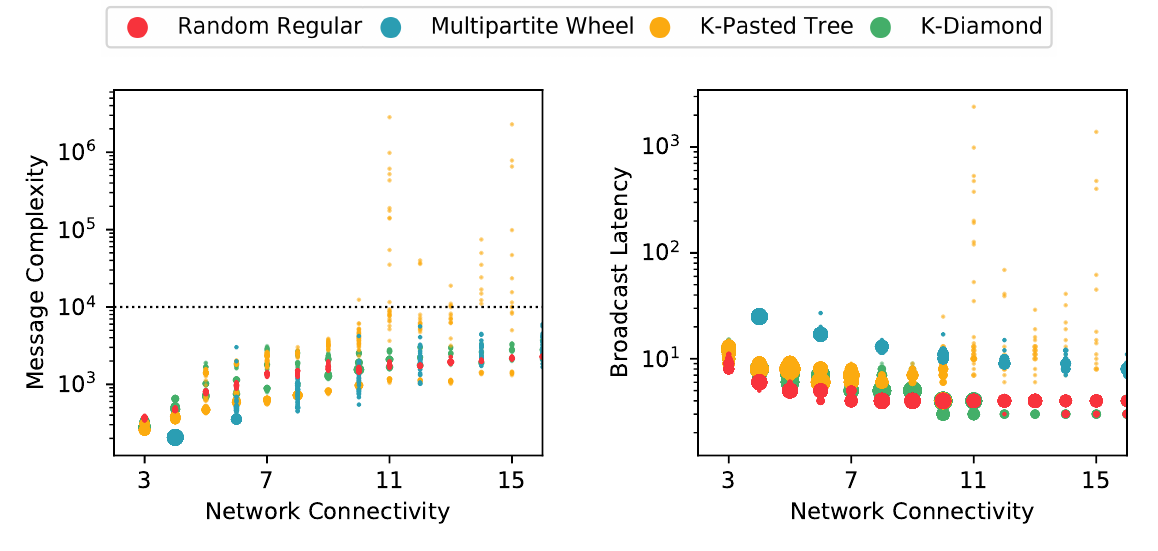}
	\caption{\textbf{Multi-Random policy} $n=100$, $f$ =$\lfloor(k -1)/2\rfloor$ passive Byzantines, bounded channels, message complexity and broadcast latency.}
	\label{fig:multirandmsg}
\end{figure}
\subsection*{Multi-shortest policy detailed evaluation}
We assume bounded channels and the Multi-Shortest policy, considering networks of size $n=150$ and $n=200$, topologies random regular, multipartite wheel, k-diamond and k-pasted-tree, passive and active Byzantine failures. First results are presented in Figures \ref{fig:multishort150all} and \ref{fig:multishort200all}.
It is possible to see that the trends of the message complexity and broadcast latency keep defined employing our protocol joined with the Multi-Shortest policy while increasing the size of the network \added{and considering passive Byzantine failures}. 
Specifically, \deleted{in both cases of passive and active Byzantine failures,} the message complexity keeps always close or below the $n^2$ boundary.
It can also be deduced that a regular network not necessarily results in optimal performances employing our protocol, indeed there are notable differences in the results obtained considering different topologies.

\begin{figure}[H]
	\centering
	\includegraphics[height=\textheight]{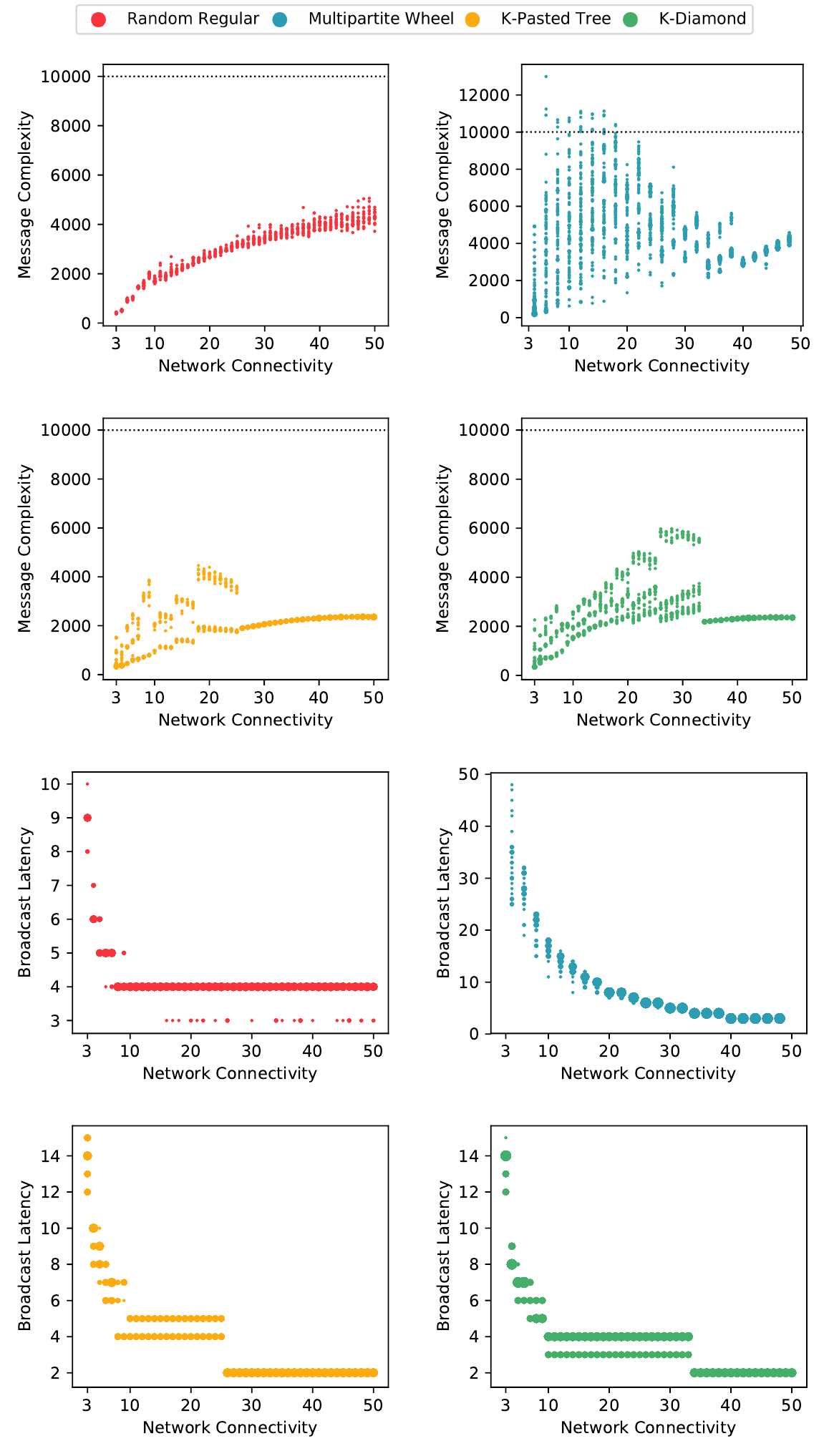}
	\caption{\textbf{Multi-Shortest policy, message complexity and broadcast latency, passive Byzantines, bounded channels, $n=100$}, $f$ =$\lfloor(k -1)/2\rfloor$.}
	\label{fig:multishort100all}
\end{figure}

\begin{figure}[H]
	\centering
	\includegraphics[height=\textheight]{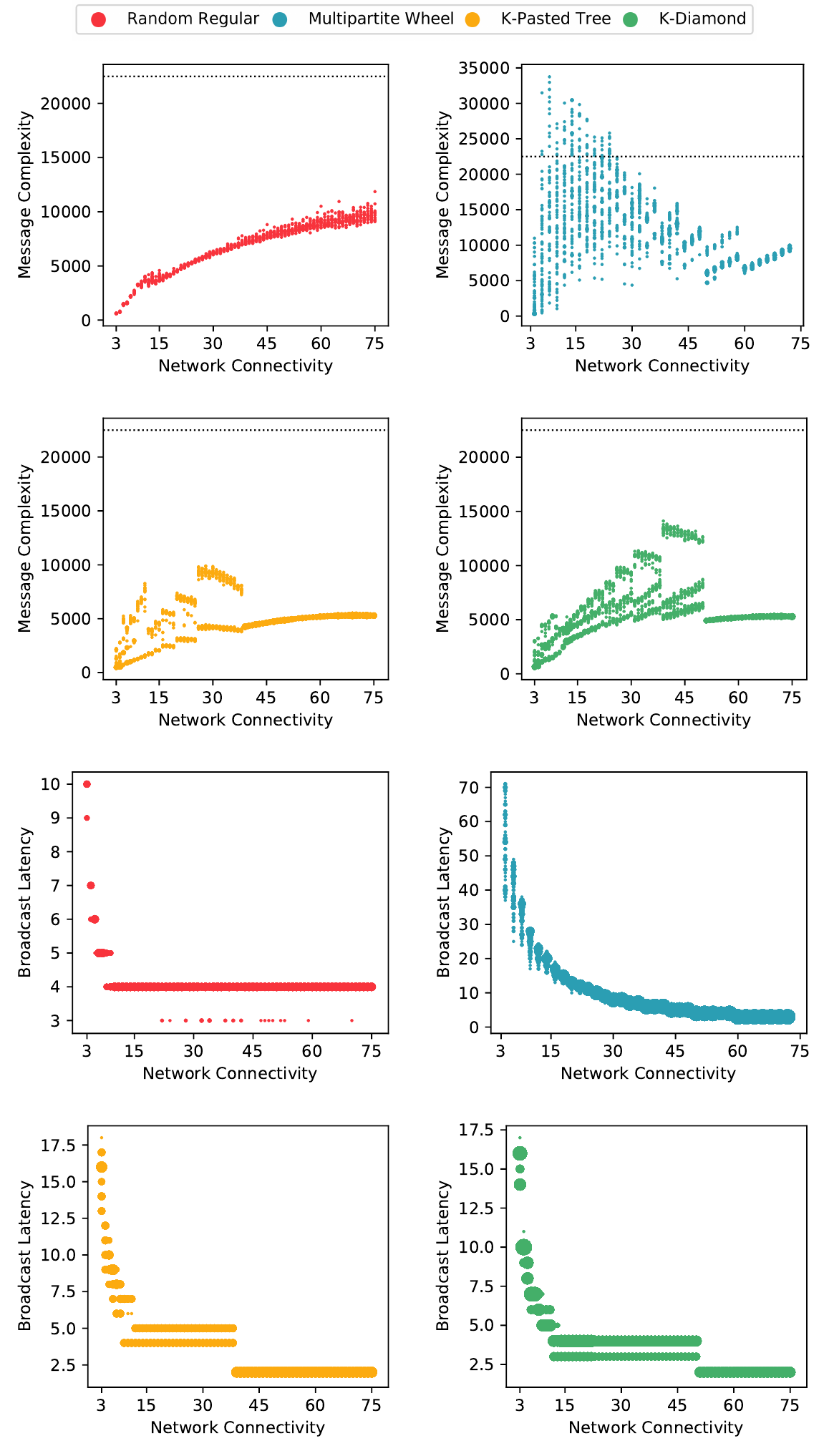}
	\caption{\textbf{Multi-Shortest policy, message complexity and broadcast latency, passive Byzantines, bounded channels}, $n=150$, $f$ =$\lfloor(k -1)/2\rfloor$.}
	\label{fig:multishort150all}
\end{figure}
\begin{figure}[H]
	\centering
	\includegraphics[height=\textheight]{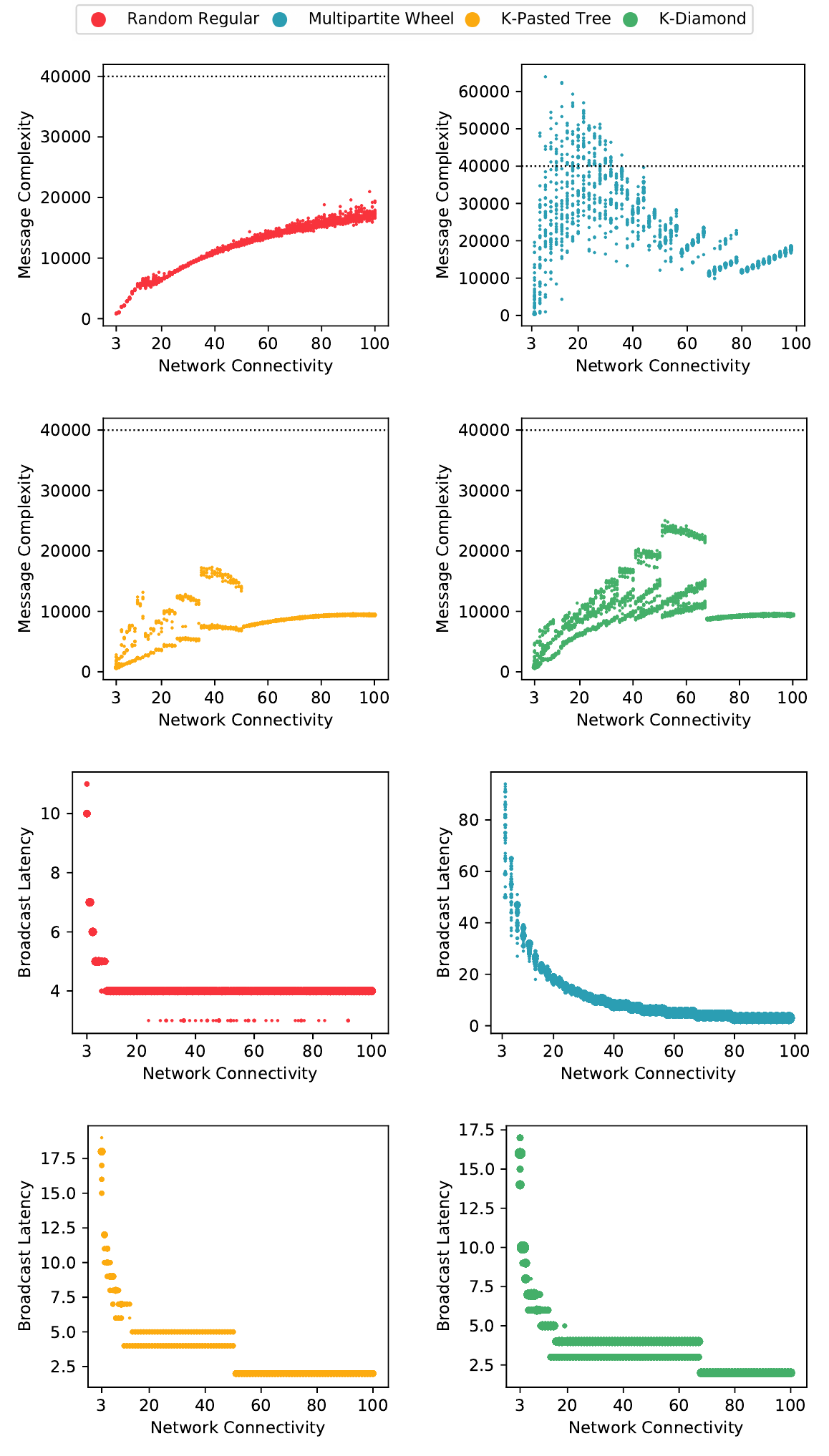}
	\caption{\textbf{Multi-Shortest policy, message complexity and broadcast latency, passive Byzantines, bounded channels}, $n=200$, $f$ =$\lfloor(k -1)/2\rfloor$.}
	\label{fig:multishort200all}
\end{figure}

\added{It can also be noticed from the distribution of the measures that there are several topologies (K-Pasted-Tree, K-Diamond and expecially Multipartite wheel)  where the placement of the source and the Byzantine failures plays a remarkable impact on the message complexity. Additional details will be later provided.}

To evaluate the effects of the Multi-Shortest policy on the broadcast latency, we simulate the \emph{BFT-BRB} protocol employing either the Multi-Shortest policy or unbounded channels, considering passive Byzantine processes and networks of size $n=100$. It can be deduced (Figure \ref{fig:latencydelay}) that the policy we defined introduces negligible delays.
\begin{figure}[!htb]
	\includegraphics[width=\textwidth]{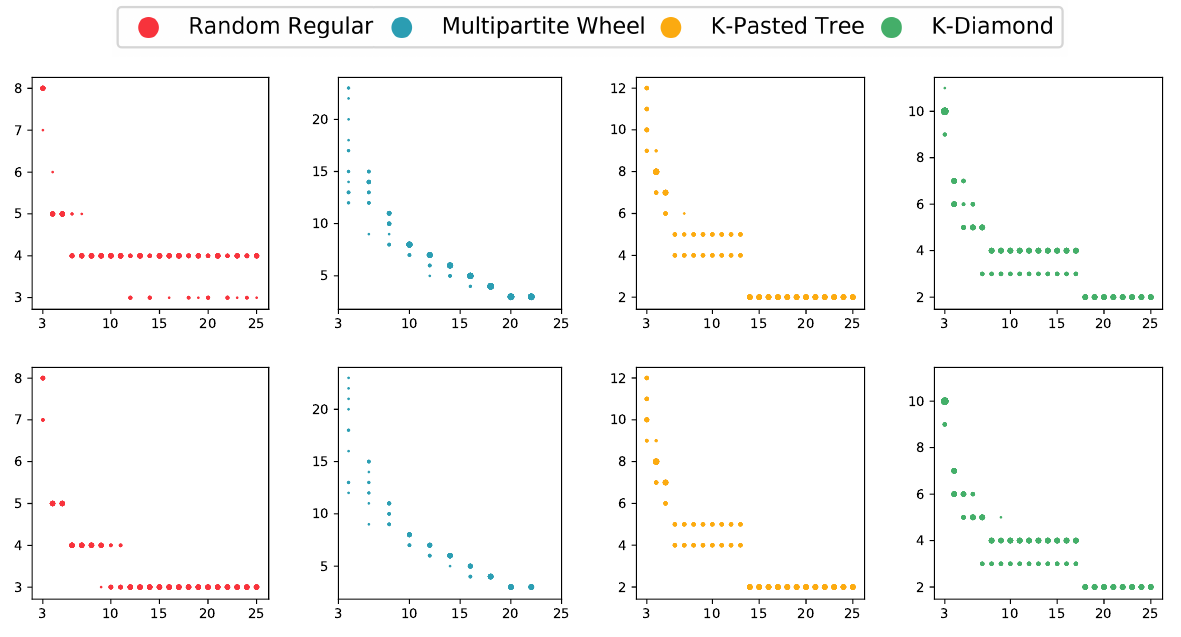}
	\caption{\textbf{Delay forwarding policy.} X-axis: network connectivity, Y-axis: broadcast latency, fist row: Multi-Shortest bounded channel, second row: unbounded channel, $n=50$, $f =\lfloor(k -1)/2\rfloor$ passive Byzantine.}
	\label{fig:latencydelay}
\end{figure}

\added{We move to consider the case of active Byzatine processes, specifically in Figure \ref{fig:multishortnonomni} general active (non omniscient) Byzantine faults are assumed. It can be noticed that, spreading spourious pathsets (using the strategy we defined) once they get knowledge about a content, the Byzantine processes have no negative impact on the message complexity. As a matter of facts they may even help correct processes achieving reliable broadcast (because they relay the content even if they try not to increase the VC on the receiving processes).}

\begin{figure}[H]
	\centering
	\includegraphics[width=.83\textwidth]{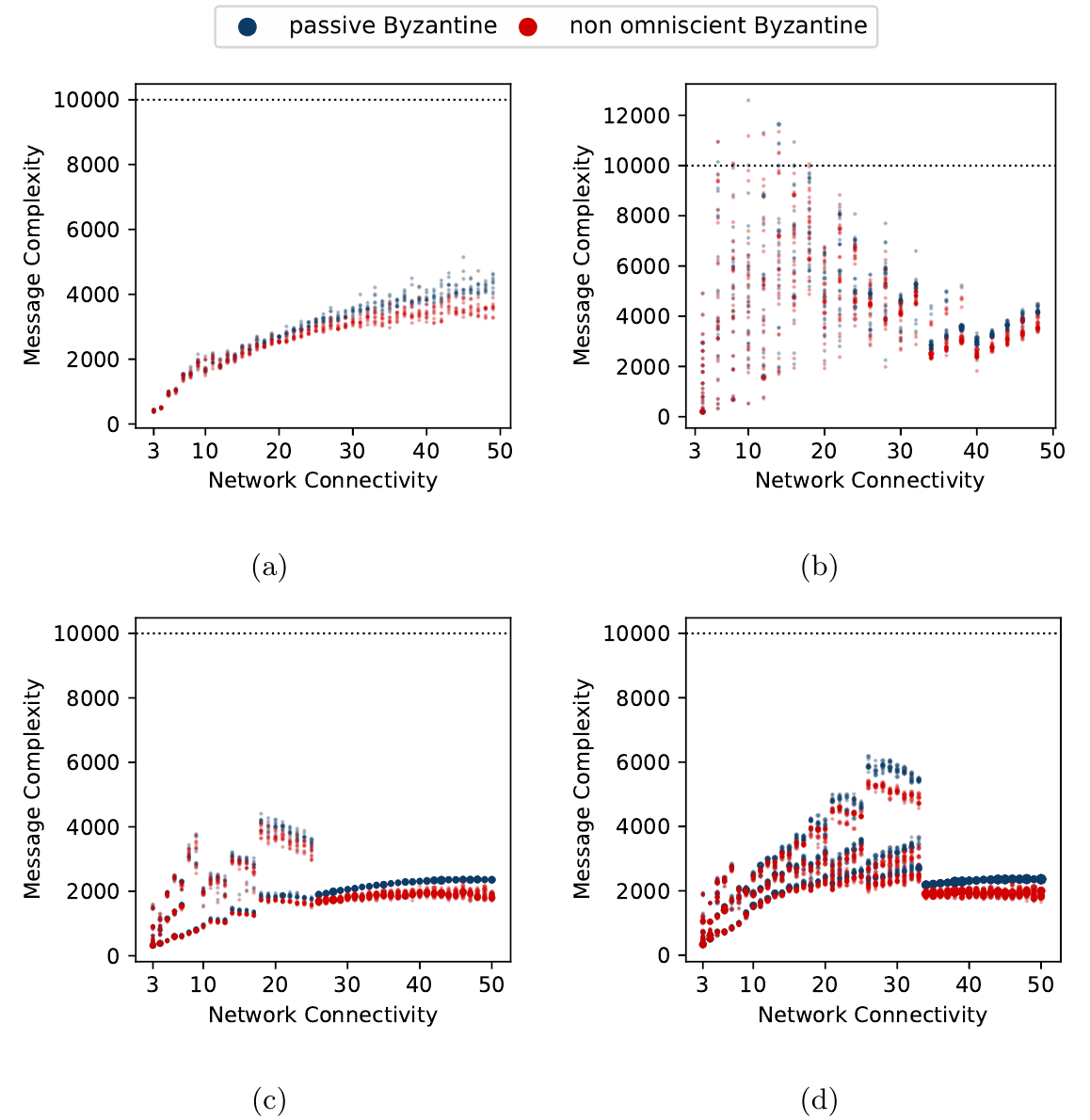}
	\caption{\textbf{Multi-Shortest policy, message complexity with either passive or non omniscient (general) Byzantine.} $n=100$, ${f =\lfloor(k -1)/2\rfloor}$, bounded channels. (a) random regular, (b) multipartite wheel (c) k-pasted-tree, (d) k-diamond.}
	\label{fig:multishortnonomni}
\end{figure}

\added{
We consider the case of omniscient Byzantine faulty processes, which start spreading spurious pathset about the content from the beginning of broadcast. The results we obtained are presented in Figures \ref{fig:multishort100byz}, \ref{fig:multishort150byz}. It is possible to see that such stronger Byzantine faults are able to remarkably increase the message complexity, nonetheless it keeps close to the $n^2$ threshold.
}

\begin{figure}[H]
	\centering
	\includegraphics[width=.8\textwidth]{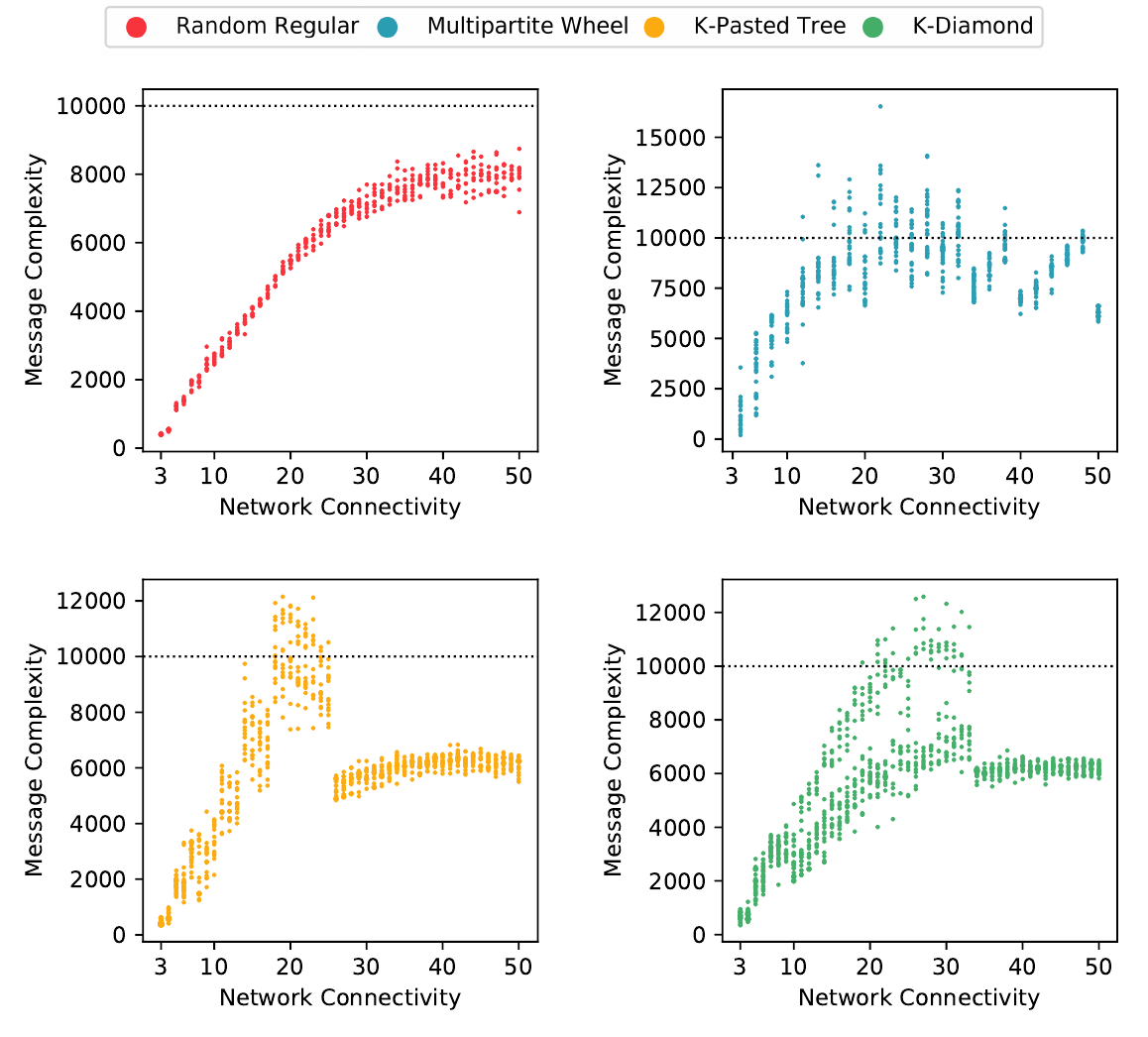}
	\caption{\textbf{Multi-Shortest policy, omniscient Byzantine faults.} $n=100$, $f =\lfloor(k -1)/2\rfloor$, bounded channels.}
	\label{fig:multishort100byz}
\end{figure}

\begin{figure}[H]
	\centering
	\includegraphics[width=.8\textwidth]{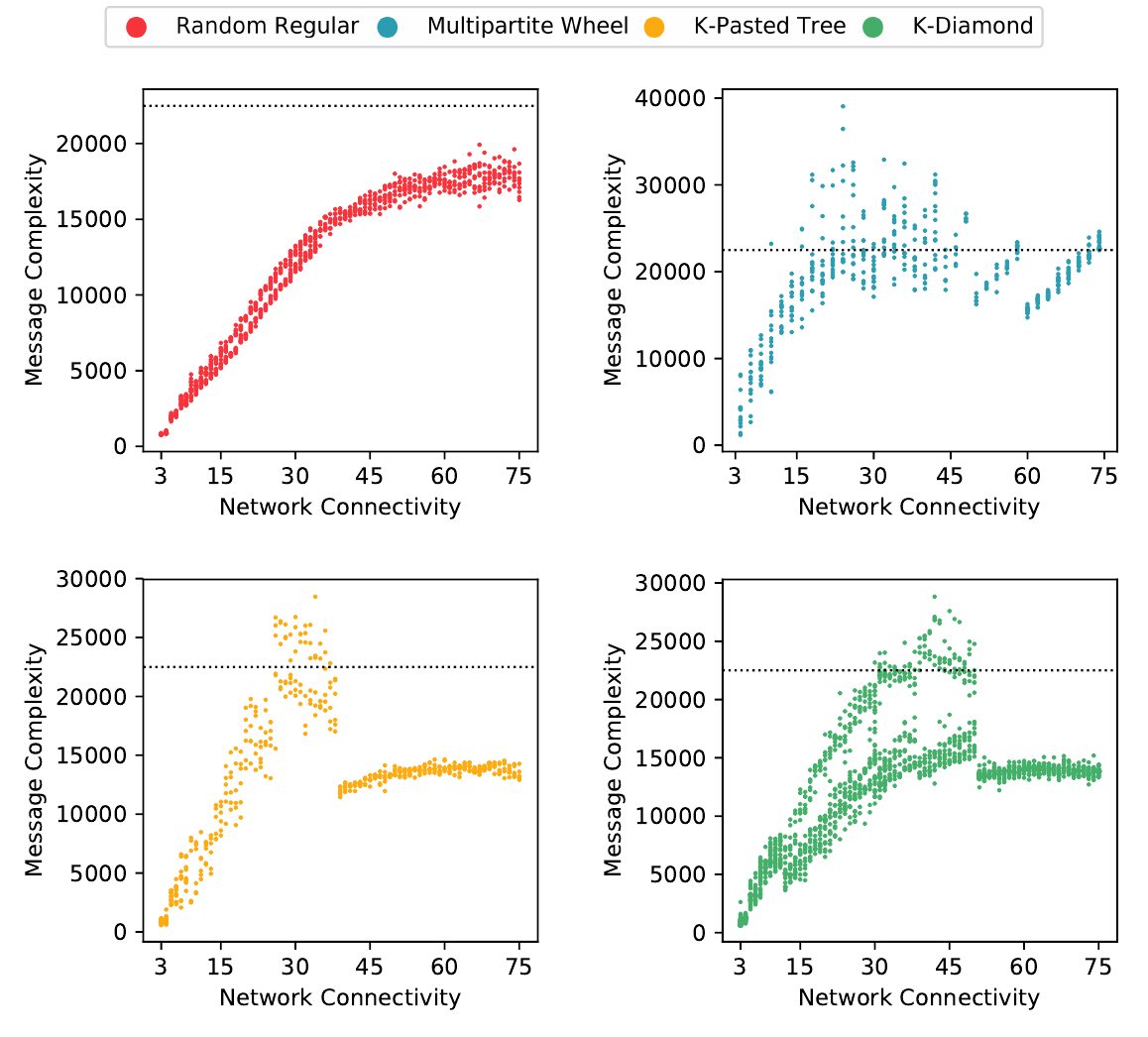}
	\caption{\textbf{Multi-Shortest policy, omniscient active Byzantine faults.} $n=150$, $f$ =$\lfloor(k -1)/2\rfloor$, bounded channels.}
	\label{fig:multishort150byz}
\end{figure}

\subsection*{Varying the number of failures}

\added{We evaluate how the message complexity evolves when the number of faulty processes is not maximized. We plotted the results we obtained in Figures \ref{fig:varyF_cut},\ref{fig:varyF_byz}. Whatever is the amout of failures, processes deliver a content only if the associated minimum cut is greater than $\lfloor(k -1)/2\rfloor$.  
	
It is possible to deduce that the resulting message complexity depends on the specific topology considered and on the degree of connectivity. Specifically, both in case of passive and omniscient active Byzantine faults, there are settings where the message complexity remains constant independently from the number of effective failures and others where the message complexity increases exponentially with the number of failures.
}
	
\begin{figure}[H]
	\centering
	\includegraphics[width=\textwidth]{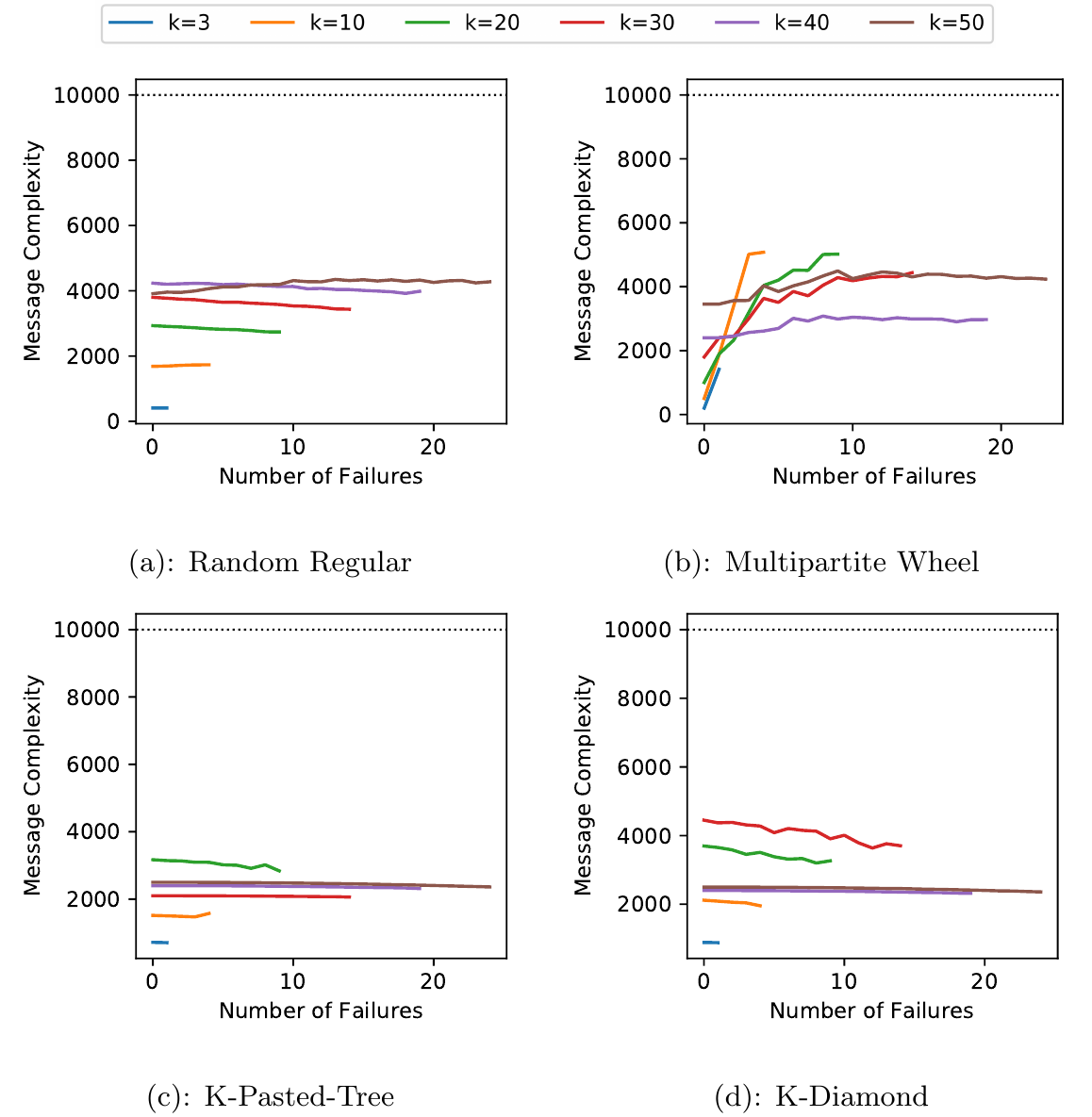}
	\caption{\textbf{Multi-Shortest policy, varying the number of faults, message complexity, passive Byzantines.} $n=100$.}
	\label{fig:varyF_cut}
\end{figure}

\begin{figure}[H]
	\centering
	\includegraphics[width=\textwidth]{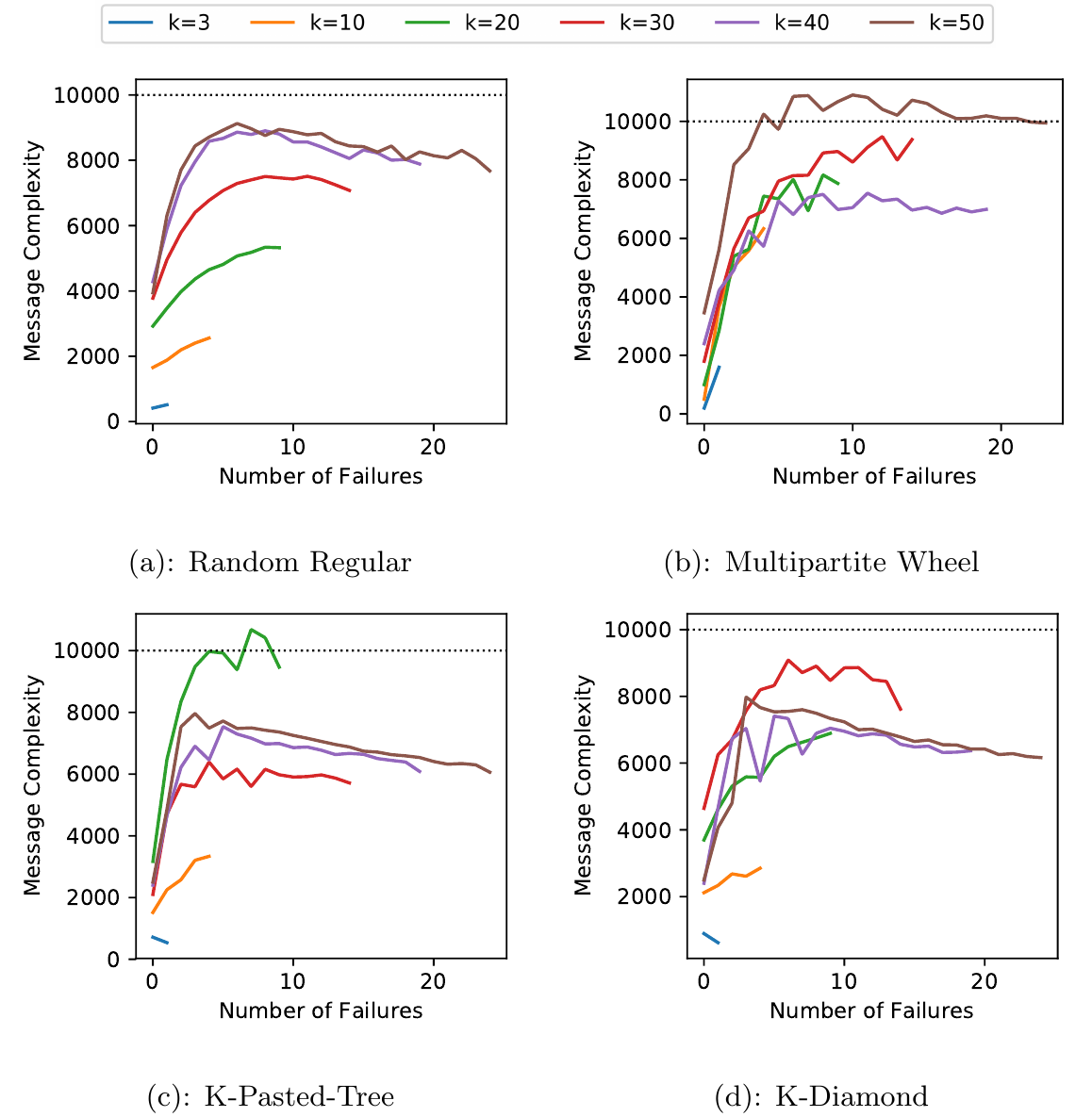}
	\caption{\textbf{Multi-Shortest policy, varying the number of faults, message complexity, active omniscient Byzantines.} $n=100$.}
	\label{fig:varyF_byz}
\end{figure}

\subsection*{Barab\'asi-Albert graph}
We separately evaluated in Figure \ref{fig:barabasimsg} our algorithm in a Barab\'asi-Albert graph while varying the attachment parameter $m$, in order analyze our protocol on a topology with different degree distribution with respect the previous analyzed. The \emph{BFT-BRB} protocol and the Multi-Shortest forwarding policy shown to keep performing in the same manner. To allow the reader to make a comparison with the other topologies, we plot in Figure \ref{fig:barabasim} the relation between the attachment parameter $m$ and the network connectivity.

\begin{figure}[!htb]
	\centering
	\includegraphics[width=.85\textwidth]{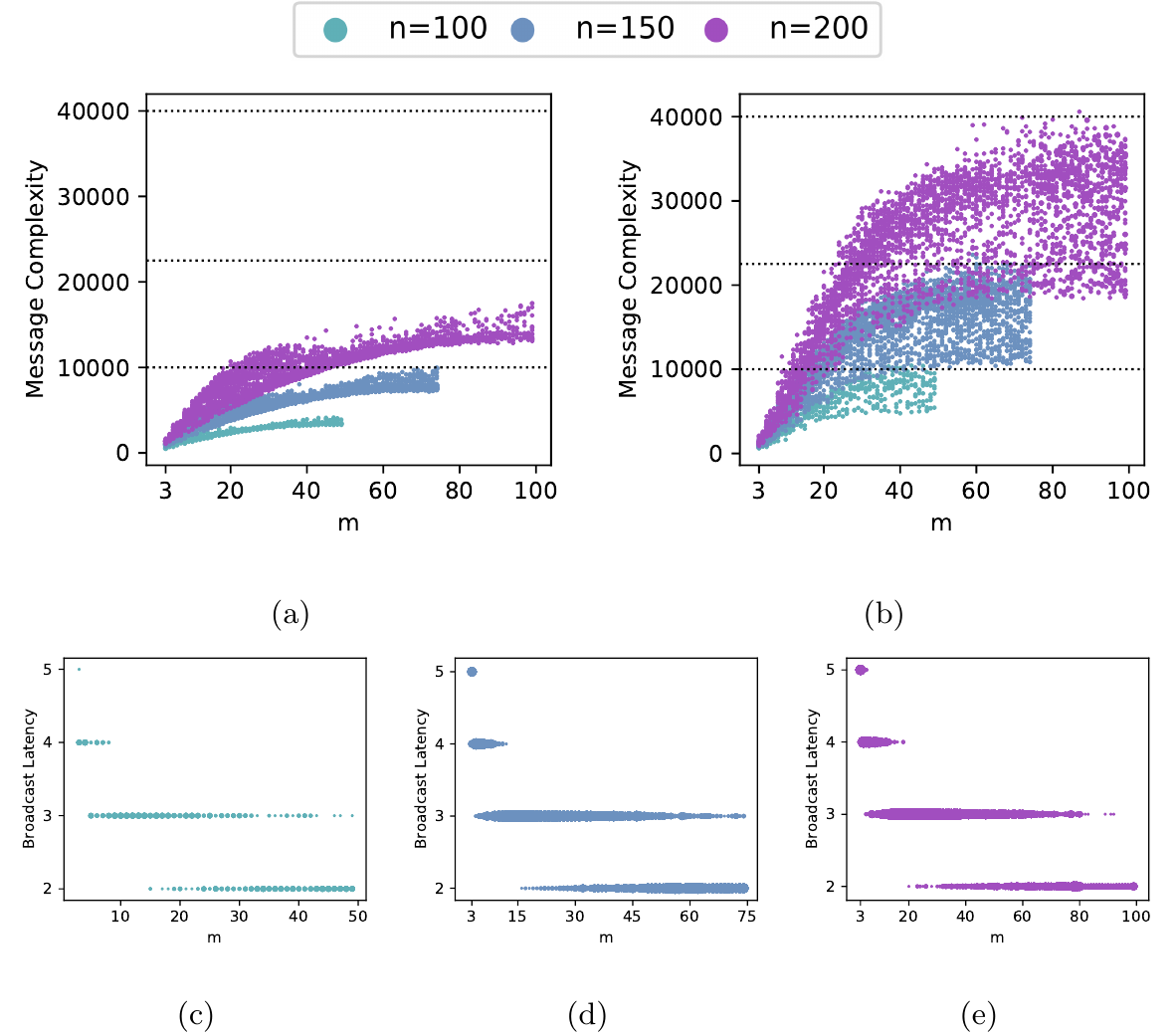}
	\caption{\textbf{Multi-Shortest policy, Barababasi-Albert Network, message complexity.}  $f$ =$\lfloor(k -1)/2\rfloor$, bounded channels, $n=100,150,200$ (a) passive Byzantines, (b) active omniscient Byzantines.}
	\label{fig:barabasimsg}
\end{figure}

\begin{figure}[!htb]
	\centering	
	\includegraphics[width=.4\textwidth]{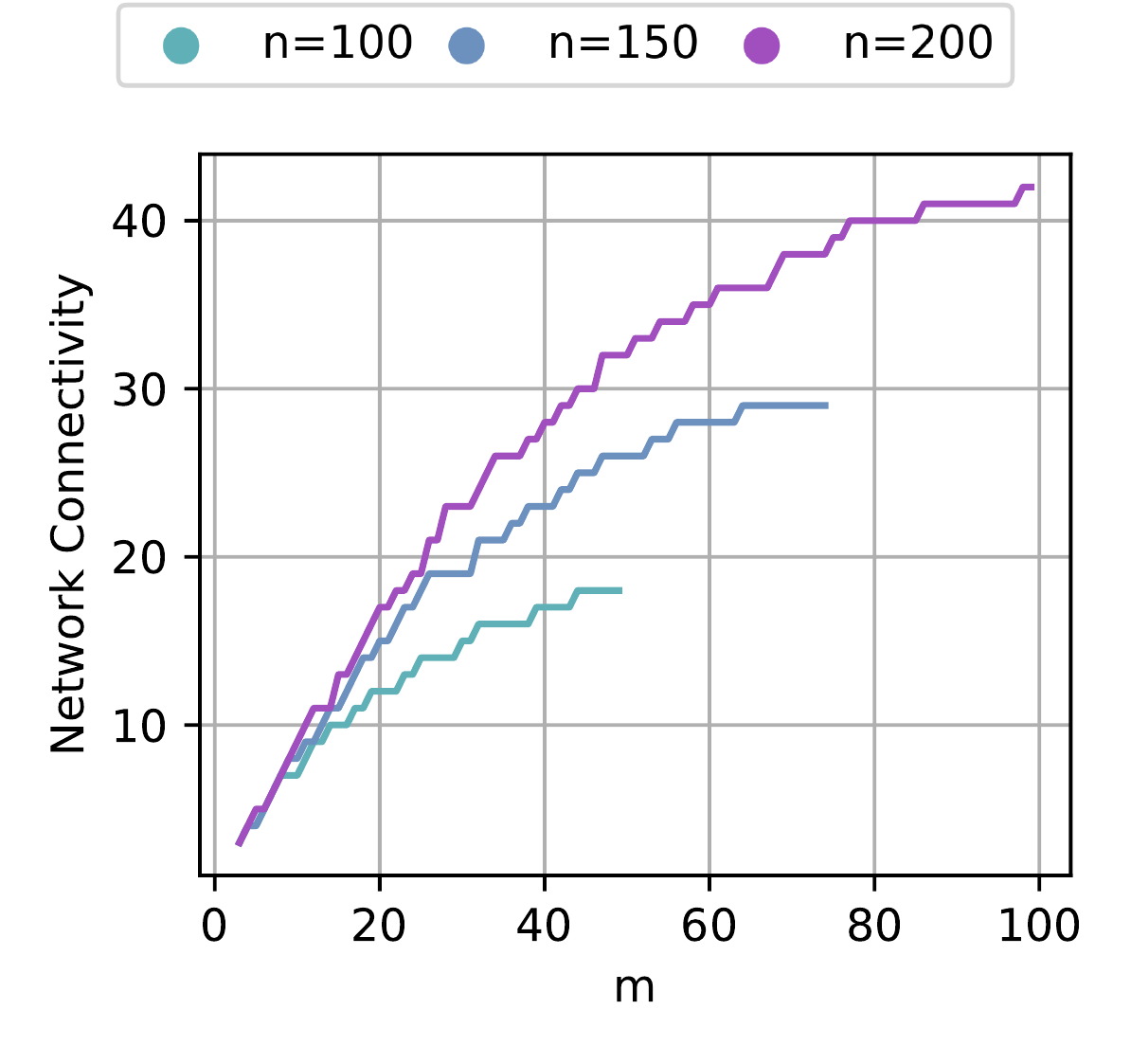}
	\caption{\textbf{Barababasi-Albert Network, relation between the attachment parameter $m$ and the network connectivity}}
	\label{fig:barabasim}
\end{figure}

These simulations allow us to conclude that a Byzantine tolerable reliable broadcast protocol practically employable in synchronous systems without considering further assumptions with respect the state of art is achievable.

\subsection*{Worst Case Scenarios}
For the ease of completeness, we briefly survey two worst case scenarios: the multipartite wheel and the generalized wheel.
In Figure \ref{fig:worstcase}a is summarized one of executions we are going to present. Let us consider the multipartite wheel of size $n=21$ and $k=6$, choose a node as source (in Figure \ref{fig:worstcase}a depicted in orange) and place two faulty processes (in red) in its neighborhood in distinct \replaced{groups (i.e. those neighbor will have different neighbors)}{sets}.  It results that only two correct processes per \replaced{group}{set} delivers the content during the first round. Subsequently, they relay the message to all the nodes in the consecutive  \replaced{group}{sets}. But, none of this node is able to deliver the message: the minimum cut of the generated paths is 2 and processes demand paths with minimum cut at least 3. The nodes succeed in delivering the message only when ``the propagation on the two sides met", achieving a minimum cut of 4. It can be noticed that a considerable amount of paths may be generated in this specific worst case scenario while the values of $n$ and $k$ increases. Nonetheless, the \emph{BFT-BRB} protocol and the Multi-Shortest policy reduced such a message complexity case as shown in Figures \ref{fig:multishort100all}, \ref{fig:multishort150all} and \ref{fig:multishort200all}. \added{We additionally simulate in Figure \ref{fig:complexworstgwheel}a our protocol with the Multi-Shortest policy on a multipartite wheel of size $n=100$ with passive Byzantine in the worst placement.}

\begin{figure}[!htb]
	\centering
	\includegraphics[width=\textwidth]{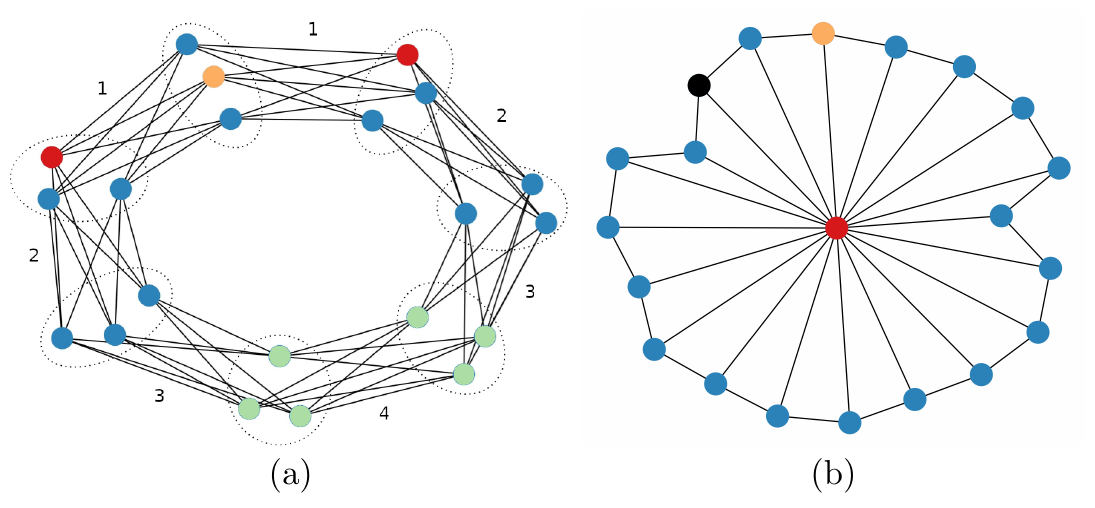}
	\caption{\textbf{Worst case scenarios.} Multi-Partite Wheel (a) and Generalized Wheel (b).}
	\label{fig:worstcase}
\end{figure}

\begin{figure}[!htb]
	\centering
	\includegraphics[width=\textwidth]{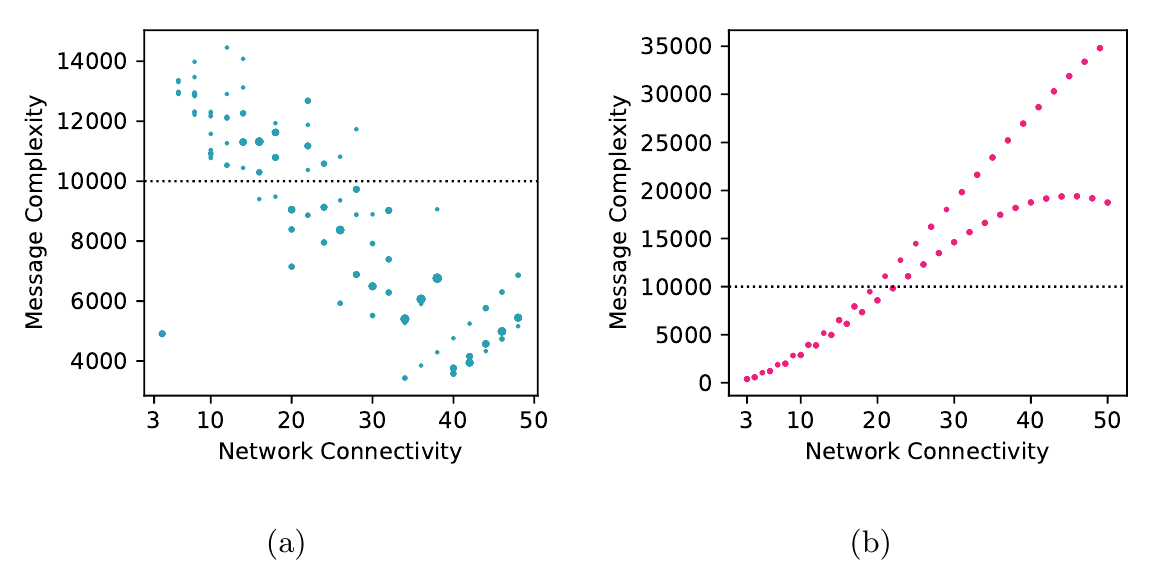}
	\caption{\textbf{Multi-Shortest policy, generalized wheel, worst Byzantine placementm passive Byzantine.} (a) multipartite wheel, (b) generalized wheel.}
	\label{fig:complexworstgwheel}
\end{figure}

Another worst case scenario is depicted in Figure \ref{fig:worstcase}b. Let us assume a generalized wheel, pick a source on the cycle and the Byzantine processes always located on the clique. The Figure \ref{fig:complexworstgwheel}b show that in this specific case our algorithm and the Multi-Shortest policy are less effective in reducing the message complexity while Byzantine processes are located in the clique.

\section*{Conclusion}
We revisited available solutions for the reliable broadcast in general network hit by up to $f$ arbitrarily distributed Byzantine failures, and proposed modifications following performance related observations. Although the delivery complexity of our protocol remains unchanged with respect the state-of-art solutions, our experiments show that it is possible to drastically reduce the message complexity (from factorial to polynomial in the size of the network), practically enabling reliable broadcast in larger systems and networks with authenticated channels.
There are several open problems that may follows: is it possible to define a solution to the hitting set problem suited for the specific input generated by our protocol? Is it possible to remove from the system the contents generated by Byzantine processes? And under which assumption? \added{Which are the graph parameters that govern the message complexity of our protocol?}
Our results open to the possibility of identifying a polynomial theoretical bound on message complexity solving the reliable broadcast problem \added{with honest dealer}.
Finally, the Bizantine Reliable Broadcast problem should be analyzed also on dynamic networks. Even if the protocol we proposed can directly be employed on asynchronous and/or dynamic systems, the achieved gain in message complexity is not guaranteed due to the weaker synchrony assumptions, and probably specific assumption on the evolution of the system must be guaranteed in searching a practical employable solution.

\bibliographystyle{spmpsci}      % mathematics and physical sciences
\bibliography{references}   % name your BibTeX data base

%\begin{thebibliography}{}
%%
%% and use \bibitem to create references. Consult the Instructions
%% for authors for reference list style.
%\end{thebibliography}
%
\end{document}